\documentclass{article}
\usepackage{color}
\usepackage[authoryear]{natbib}
\usepackage[utf8]{inputenc} % allow utf-8 input
\usepackage[T1]{fontenc}    % use 8-bit T1 fonts
\usepackage{hyperref}       % hyperlinks
\usepackage{url}            % simple URL typesetting
\usepackage{booktabs}       % professional-quality tables
\usepackage{amsfonts}       % blackboard math symbols
\usepackage{nicefrac}       % compact symbols for 1/2, etc.
\usepackage{microtype}   
\usepackage{enumitem}
\usepackage{amsmath, parallel}
\usepackage{fourier}
\usepackage{amsthm,amssymb, xspace}
\usepackage{graphicx, color}
\usepackage[ruled,vlined]{algorithm2e}
\usepackage{times}
\theoremstyle{definition}

\newtheorem{theorem}{Theorem}
\newtheorem{lemma}{Lemma}
\newtheorem{definition}{Definition}
\newtheorem{proposition}{Proposition}
\newtheorem{remark}{Remark}
\usepackage{mathrsfs}  
\usepackage{xcolor}
\theoremstyle{definition}

\begin{document}

\title{Accelerated Algorithms for Convex and Non-Convex  Optimization on Manifolds}
\author{ Lizhen Lin$ ^{1} $, Bayan Saparbayeva$^{2} $, Michael Minyi Zhang$^{3} $,\\ 
	David B. Dunson$^{4} $ }
\date{\texttt{lizhen.lin@nd.edu, bayan\_saparbayeva@urmc.rochester.edu, mzhang18@hku.hk, dunson@duke.edu}\\
$^1$ Department of Applied and Computational Mathematics and Statistics, University of Notre Dame, Notre Dame, IN, USA.\\ 
$^2$ Department of Biostatistics and Computational Biology, University of Rochester, Rochester, NY, USA.\\
$^3$ Department of Statistics and Actuarial Science, University of Hong Kong, Hong Kong, China.\\
$^4$ Department of Statistical Science, Duke University, Durham, NC, USA.\\
\today}

\maketitle

\begin{abstract}
We propose a general scheme for solving convex and non-convex optimization problems on manifolds.  The central idea is that, by adding a multiple of the squared retraction distance to the objective function in question,  we ``convexify'' the  objective function and solve a series of convex sub-problems in the optimization procedure. One of the key challenges for optimization on manifolds is the difficulty of verifying the complexity of the objective function, e.g., whether the objective function is convex or non-convex, and the degree of non-convexity. Our proposed algorithm adapts to the level of complexity in the objective function. We show that when the objective function is convex, the algorithm provably converges to the optimum and leads to accelerated convergence. When the objective function is non-convex, the algorithm will converge to a stationary point. Our proposed method unifies insights from Nesterov's original idea for accelerating gradient descent algorithms with recent developments in optimization algorithms in Euclidean space. We demonstrate the utility of our algorithms on several manifold optimization tasks such as estimating intrinsic and extrinsic Fr\'echet means on spheres and low-rank matrix factorization with Grassmann manifolds applied to the Netflix rating data set. 
\end{abstract}

\section{Introduction}
\label{sec-intro}
Optimization is a near ubiquitous tool used in a wide-range of disciplines including the physical sciences, applied mathematics, engineering and the social sciences. Formally, it aims to maximize or minimize some quantitative criteria, namely, the objective function with respect to some parameters of interest. 
In many broad, complex learning %or decision making problems arising broadly 
in modern data science, the parameters are naturally defined over %or constrained 
to be on a \emph{manifold}. The emerging field of statistics on manifolds based on Fr\'echet means \citep{rabibook, linclt} can be viewed as one of the notable examples of optimization on general manifolds.

Another example can be found in building scalable recommender systems where extracting a low-rank matrix involves an optimization problem over a Grassmann manifold \citep{boumal}. Recent development in geometric deep learning, where the input or output layer constrained to be on a Riemannian manifold \citep{Lohit2017LearningIR, Huang2017ARN, Huang2017DeepLO}, constitutes another important class of applications. Other applications arise in diverse areas ranging from medical imaging analysis, Procrustes shape matching, dimension reduction, dynamic subspace tracking, and cases involving ranking and orthogonality constraints--among many others. This proliferation of manifold-valued applications demands fundamental development of models, algorithms and theory for solving optimization problems over non-Euclidean spaces.
%One particularly successful example can be found in building scalable recommendation systems where learning a low-rank matrix factorization involves an optimization problem over a Grassmann manifold. 

% for threat detection in national security and fast response to complex events such as `cyber attacks', recommending actions or filtering important intelligence information, in particular when the input data is very big;  and (2) applications of artificial  intelligence by developing geometric deep learning techniques, at the heart  of which is an optimization problem over manifolds.
%This on one hand is motivated by the increasing prevalence of manifold-valued data in medical imaging, machine vision, and so on, as well as the explosive development for the field of big data analysi in which the 

The current literature on optimization over manifolds mainly focuses on extending existing Euclidean space algorithms, such as  Newton's method \citep{2014arXiv1407.5965S, ring12}, conjugate gradient descent \citep{Edelman1998, Nishimori2008}, steepest descent \citep{steepest}, trust-region methods \citep{Absil2007, NIPS2011_4402} and others.  Many of the objective functions in manifold optimization problems are very complex. One of key challenges for solving such problems lies in the difficulty in verifying the convexity and the degree of convexity of the objective function.
%complexity of the objective function, e.g., whether the objective function is convex, or non-convex and the degree of non-convexity.
%Current approaches lack the ability to accommodate the complexity of the optimizations. 
Current approaches cannot adapt to the complexity of the problem at hand in manifold spaces. 
%,  which is a major hurdle for optimization on manifolds as it is often challenging to verify  convexity or non-convexity of a complex objective function.  
%However, a critical gap exists in developing algorithms for non-convex optimizations on manifolds. 

We take a major step to address these issues by proposing a general scheme to solve convex and non-convex optimization problems on manifolds using gradient-based algorithms originally designed for convex functions.
%by leveraging some recent developments for optimization problems in the Euclidean space \citep{pmlr-v84-paquette18a}.  
 The key idea is to ``convexify'' the objective function by adding a multiple of the squared retraction distance.  The proposed algorithm does not require knowledge of whether the objective function is convex but will automatically converges to an optima if the function is strongly convex. When the objective is non-convex, it achieves rapid convergence to a stationary point. The proposed algorithm is a generalization of Nesterov acceleration \citep{Nesterov2004}, which improves the convergence rate of gradient descent algorithms. Our algorithm (which we call $\mathcal{A}_2$) takes any general existing optimization method (which we call $\mathcal A$), originally designed for convex functions, and converts it into a method applicable for non-convex functions. % that applies to non-convex functions. 

Similar schemes have been explored for optimization problems in Euclidean space \citep{pmlr-v84-paquette18a}. Generalizations to arbitrary manifolds, however, require fundamentally novel theoretical development.  In the Euclidean case, the gradient steps are taken towards lines, whereas for the manifold we use the retraction curves which crucially
 affects the result and raises the difficulty in proving convergence. Also for manifolds, it is not trivial to correctly  convexify the 'weakly-convex function', a broad class of non-convex functions on manifolds  we consider which  account for most of the interesting examples of non-convex functions in machine learning. We propose a novel idea to convexify the objective locally with the help of the retraction. Key features of our algorithm include adaptation to the unknown weak convexity of the objective function  and automatic Nesterov acceleration. The proposed algorithm can be used to accelerate a broad class of $\mathcal A$ algorithms including gradient descent as well as parallel optimization approaches \citep[see][]{lizhennips2018}.

Our paper is organized as follows: In Section~\ref{sec:related}, we introduce related work on accelerated optimization algorithms. Next, we present our proposed acceleration algorithm on manifolds in Section~\ref{sec: proposed} and present theoretical convergence results. In Section~\ref{sec-simu}, we consider a simulation study of estimating Fr\'echet means and a real data example using the Netflix prize data set in a matrix completion problem.

\section{Related work}
\label{sec:related}

\cite{NIPS2017_7072} propose accelerated first-order methods for \emph{geodesically convex optimization} on Riemannian manifolds. This is a direct generalization of Nesterov's original linear extrapolation mechanism to general Riemannian manifolds via a non-linear operator. One drawback of \cite{NIPS2017_7072} is that the accelerated step of their algorithm  involves exact solving of non-trivial implicit equations.
% (see their equations 4 and 5).
% which requires computing  parallel transport on the manifold, a  computational-expensive task for many non-flat manifolds. 

\cite{pmlr-v75-zhang18a} later proposed a computationally tractable accelerated gradient algorithm and a novel estimation sequence for convergence analysis. Our approach is fundamentally different from theirs. We regularize an objective function with a squared retraction distance (see Proposition~\ref{prop1}), solve a sequence of convex subproblems, adapt to the degree of weak convexity of the objective function, and produce accelerated rates for convex objectives. 
%One of the fundamental difference between their method and  ours is that that our algorithms can handle both convex and non-convex objective functions, with convergence guarantee. 
Even in the convex case, our approach can deal with a much broader class of retraction-based convex functions. %In practice, it is highly non-trivial task to derive the degree of non-convexity of a function on manifold. In addition, even for the case of convex objective functions, our methods provides acceleration for general retraction-based convex functions, which is a much broader class than the geodesic-based counter parts. 

\cite{pmlr-v84-paquette18a} proposes a general scheme called ``Catalyst acceleration'' for solving general optimizations  in  Euclidean space, which has inspired development of  some ideas for our work.  Similar ideas have been explored  for convex functions in Euclidean space in both theory and practice \citep{Lin:2017}. However, optimization problems on manifolds are of fundamentally different nature and require development of  substantially new tools and theory. 

There is an interesting line of work on proposing fast  algorithms for stochastic optimization on manifolds \citep[see][]{Zhang:2016:RSF, 2018arXiv181104194Z, pmlr-v89-zhou19a, 6487381} which %in general
employ very different techniques such as minibatching, variance reduction and  utilizing the uncertainty of inputs. Methods like \cite{zhang2018towards} propose optimization methods that are analogous to Nesterov-type algorithms for manifold spaces.

\section{Accelerated algorithms for  optimization on manifolds}
\label{sec: proposed}
\subsection{Weakly convex functions on  manifolds with respect to retraction mapping}

We first define general retraction-based, weakly convex, convex, and strongly convex functions by generalizing from their geodesic-based counterparts . We then prove an important proposition 
that can transform a non-convex function into a convex one simply by adding a multiple of the squared retraction-distance to the objective function. 

\begin{definition}
A {\it retraction} on a manifold $\mathcal{M}$ is a smooth mapping from its tangent bundle $\mathcal{R}: T\mathcal{M}\rightarrow\mathcal{M}$
with the following properties:
\begin{enumerate}
    \item $\mathcal{R}_{\theta}(0_{\theta})=\mathcal{R}(\theta, 0_{\theta})=\theta,$ where $0_{\theta}$ denotes the zero vector on the tangent space $T_{\theta}\mathcal{M};$
    \item For any point $\theta\in\mathcal R$ the differential $d(\mathcal{R}_{\theta})$ of the retraction mapping at the zero vector $0_{\theta}\in T_{\theta}\mathcal M$ has to be equal to  the identity mapping on $T_{\theta}\mathcal{M},$ that is $d(\mathcal{R}_{\theta}(0_{\theta}))=d\big(\mathcal{R}(\theta, 0_{\theta})\big)={\rm id}_{T_{\theta}\mathcal{M}},$ where ${\rm id}_{T_{\theta}\mathcal{M}}$ denotes the identity mapping on $T_{\theta}\mathcal{M}.$ 
\end{enumerate}
The exponential map on a Riemannian manifold can be viewed as a special case of the retraction map, and the inverse-exponential map is a special case of the inverse-retraction map. A good choice of retraction map can lead to substantial reduction in computation burden compared to the exponential map. We see an example in Section~\ref{subsec:netflix} on the choice of a retraction map for Grassmannian; Figure~\ref{fig-ret} provides a visualization of a retraction map. 
\begin{figure}
\center
\includegraphics[width=.5\linewidth]{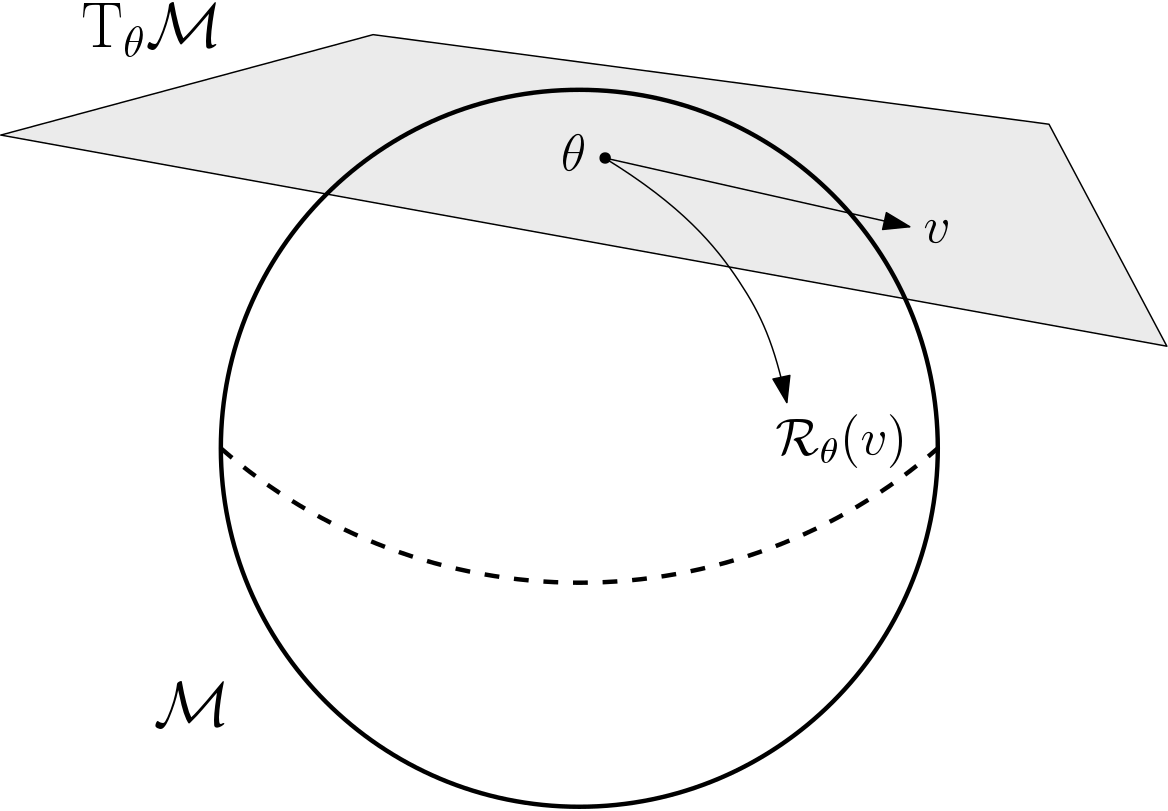}
\caption{Illustration of a retraction map on a manifold}
\label{fig-ret}
\end{figure}
\end{definition}
%We require the following symmetry properties (generalizing Riemannian symmetric space):
%\begin{enumerate}
  %  \item For any $\theta_1, \theta_2\in\mathcal{M},$ curves $\mathcal{R}_{\theta_1}\left(t_1\mathcal{R}_{\theta_1}^{-1}\theta_2\right)$ and $\mathcal{R}_{\theta_2}\left(t_2\mathcal{R}_{\theta_2}^{-1}\theta_1\right),$ where $t_1, t_2\in[0, 1],$ must coincide;
    %\item For any $\theta_1, \theta_2\in\mathcal{M}$ the equality 
    %\begin{align*}
   % $
     %   d_{\mathcal{R}}(\theta_1, \theta_2)=\|\mathcal{R}_{\theta_1}^{-1}\theta_2\|=\|\mathcal{R}_{\theta_2}^{-1}\theta_1\|
    %$
   %holds, where  $d_{\mathcal{R}}(\cdot, \cdot)$ denotes the retraction distance on $\mathcal M$.
%\end{enumerate}

We first define the retraction distance function on $\mathcal M$
\begin{equation*}
 d_{\mathcal{R}}(\theta_0, \theta)=\|\mathcal{R}_{\theta_0}^{-1}\theta\|.
 \end{equation*}

Since at  zero the differential of the retraction map is the identity, there is a small enough neighborhood  $D$ of the point $\theta$ where the inverse retraction map $\mathcal{R}_{\theta}^{-1}$ is bi-Lipschitz continuous in $D,$ i.e. $d_{\mathcal{R}}$  satisfies inequalities
\begin{equation*}
    \frac{1}{K_1}d_{\mathcal{R}}(\vartheta_1, \vartheta_2)\leq\|\mathcal{R}_{\theta}^{-1}\vartheta_1-\mathcal{R}_{\theta}^{-1}\vartheta_2\|\leq K_2d_{\mathcal{R}}(\vartheta_1, \vartheta_2),
\end{equation*}
where $\vartheta_1, \vartheta_2\in D,$ and $K_1\geq1,$ and $K_2\geq1.$

In addition, we also require the squared retraction distance function to be {\it $2R_1$-strongly retraction convex} around $\vartheta$--that is, for some $\delta>0$ and constant $0\leq R_1\leq1$ the following inequality holds:
\begin{equation}\label{dR-strongly-convex-2}
    d_{\mathcal{R}}^2(\theta_2, \vartheta)\geq d_{\mathcal{R}}^{2}(\theta_1, \vartheta)+\langle\nabla d_{\mathcal{R}}^2(\theta_1, \vartheta), \mathcal{R}_{\theta_1}^{-1}\theta_2\rangle+R_1d_{\mathcal{R}}^2(\theta_1, \theta_2),
\end{equation}
where $d_{\mathcal{R}}(\theta_i, \vartheta)<\delta,$ $i=1,2.$  Due to the fact that at the zero vector $0_{\vartheta}\in T_{\vartheta}\mathcal{M}$ the differential of $\mathcal{R}_{\vartheta}$  is equal to identity mapping, we can see that in a small neighborhood of $\vartheta$,  the square retraction distance function behaves like the square normal function which is strongly convex.

\begin{definition}
Consider a function $f:\mathcal{M}\rightarrow\bar{\mathbb{R}}$ and a point $\theta$ with $f(\theta)$ finite. The {\it $\mathcal{R}-$subdifferential} of $f$ at $\theta$ is the set
\begin{equation*}
\begin{split}
    \partial f(\theta)=\Big\{v\in T_{\theta}\mathcal{M}: f(\vartheta)\geq f(\theta)+\langle v, \mathcal{R}^{-1}_{\theta}\vartheta\rangle+o\big(d_{\mathcal{R}}(\theta, \vartheta)\big)\\
     \forall \vartheta\in \mathcal{M}\Big\}.
    \end{split}
\end{equation*}
\end{definition}

We now define the notion of convex functions on manifolds with respect to the retraction map. 
\begin{definition}
A function $f$ is {\it convex with respect to the retraction} $\mathcal{R}$ if for any points $\theta_1, \theta_2\in\mathcal{M}$ the inequality holds
\begin{equation}
\label{f-convex-2}
    f(\theta_2)\geq f(\theta_1)+\langle v, \mathcal{R}_{\theta_1}^{-1}\theta_2\rangle, \qquad  v\in\partial f(\theta_1).
    %-\frac{\rho}{2}d_{\mathcal{R}}^2(\theta_1, \theta_2), 
\end{equation}
\end{definition}

Now we are ready to define one of the most important classes of non-convex functions called  \emph{weakly-convex functions} which constitute many interesting applications of non-convex functions in machine learning. 

\begin{definition}
A function $f$ is {\it $\rho$-weakly convex with respect to the retraction} $\mathcal{R}$ if for any points $\theta_1, \theta_2\in\mathcal{M}$ the inequality holds
\begin{equation}
\label{f-weakly-convex-2}
    f(\theta_2)\geq f(\theta_1)+\langle v, \mathcal{R}_{\theta_1}^{-1}\theta_2\rangle-\frac{\rho}{2}d_{\mathcal{R}}^2(\theta_1, \theta_2), \qquad  v\in\partial f(\theta_1).
\end{equation}
\end{definition}

Given the strong retraction convexity of the squared retraction distance (see \eqref{dR-strongly-convex-2}),  we can regularize the weakly convex function $f$ by adding the term $\frac{\kappa}{2}d_{\mathcal{R}}^2(\theta, \vartheta)$ and turn it into a convex function  through the following proposition. 
%We first prove the following important proposition which says for any $\rho$-weakly convex function, one can add a multiple of the squared distance to make it convex, which allows one to turn a non-convex optimization into a convex one. 
\begin{proposition}
\label{prop1}
Let $d_{\mathcal{R}}$  be a retraction distance that is  strong-retraction convex or  satisfies the inequality \eqref{dR-strongly-convex-2} in the subset $D\subset\mathcal{M}.$ 
Then the function $f$ is $R_1\kappa$-weakly convex in $D$ if and only if the function
     \begin{equation*}h_{\kappa}(\theta, \vartheta)=f(\theta)+\frac{\kappa}{2} d_{\mathcal{R}}^2(\theta, \vartheta)\end{equation*}
is convex in $D.$
\end{proposition}
\begin{proof}
Let $f$ be $\rho$-weakly convex. Then for any $\theta_1, \theta_2\in D$ and any $\lambda\in[0, 1]$
\begin{align*} 
 f(\theta_2)&\geq f(\theta_1)+\langle\partial f(\theta_1), \mathcal{R}_{\theta_1}^{-1}\theta_2\rangle-\frac{R_1\kappa}{2}d_{\mathcal{R}}^2(\theta_1, \theta_2)\\
 &\geq f(\theta_1)+\langle\partial f(\theta_1), \mathcal{R}_{\theta_1}^{-1}\theta_2\rangle+ \frac{\kappa}{2}d_{\mathcal{R}}^{2}(\theta_1, \vartheta)\\
& \quad\quad +\langle\nabla \frac{\kappa}{2}d_{\mathcal{R}}^2(\theta_1, \vartheta), \mathcal{R}_{\theta_1}^{-1}\theta_2\rangle-\frac{\kappa}{2}d_{\mathcal{R}}^2(\theta_2, \vartheta),
\end{align*}
which implies
\begin{equation*}
    h_{\kappa}(\theta_2, \vartheta)\geq h_{\kappa}(\theta_1, \vartheta)+\langle\partial h_{\kappa}(\theta_1, \vartheta), \mathcal{R}_{\theta_1}^{-1}\theta_2\rangle.
\end{equation*}
\end{proof}
For functions defined on an Euclidean space we have a definition of a weakly convex function that is equivalent to \eqref{f-weakly-convex-2}:
\begin{equation}
\label{f-weakly-convex-0}
    f\big(\mathcal{R}_{\vartheta}(\lambda\mathcal{R}^{-1}_{\vartheta}\theta)\big)\leq\lambda f(\theta)+(1-\lambda)f(\vartheta)+\frac{\rho\lambda(1-\lambda)}{2}d_{\mathcal{R}}^2(\theta, \vartheta).
\end{equation}
Over the manifold,  however,  there is no such straightforward equivalence. This is due to the distance function $d_{\mathcal{R}}^2(\vartheta, \theta)$ which does not satisfy the following equality:
\begin{align}
\label{dR-equality}
     d_{\mathcal{R}}^2(\mathcal{R}_{\theta_1}\lambda\mathcal{R}_{\theta_1}^{-1}\theta_2, \vartheta)&\nonumber=\lambda d_{\mathcal{R}}(\theta_2, \vartheta)+(1-\lambda)d_{\mathcal{R}}^2(\theta_1, \vartheta)\\
     &\quad\quad-\lambda(1-\lambda)d_{\mathcal{R}}^2(\theta_1, \theta_2).
\end{align}
Nevertheless in some neighborhood of $\vartheta,$, for some $\delta>0$, the following inequality holds
\begin{equation}\label{dR-strongly-convex-1}
\begin{split}
   d_{\mathcal{R}}^2(\mathcal{R}_{\theta_1}\lambda\mathcal{R}_{\theta_1}^{-1}\theta_2, \vartheta)\leq\lambda d_{\mathcal{R}}^2(\theta_2, \vartheta) +(1-\lambda)d_{\mathcal{R}}^2(\theta_1, \vartheta)\\
   -\lambda(1-\lambda)R_1 d_{\mathcal{R}}^2(\theta_1, \theta_2),
   \end{split}
\end{equation}
where $d_{\mathcal{R}}(\theta_1, \vartheta)<\delta$ and $d_{\mathcal{R}}(\theta_2, \vartheta)<\delta$.  
%where $r_1\leq1.$

Therefore the function $f$ is {\it $\rho$-weakly convex with respect to the retraction} $\mathcal{R}$ if for any points $\theta, \vartheta\in\mathcal{M}$ such that $\lambda\in[0, 1],$ the approximate secant inequality holds
\begin{equation*}\label{f-weakly-convex-1}
    f\big(\mathcal{R}_{\vartheta}(\lambda\mathcal{R}^{-1}_{\vartheta}\theta)\big)\leq\lambda f(\theta)+(1-\lambda)f(\vartheta)+
    \frac{\rho\lambda(1-\lambda)}{2}d_{\mathcal{R}}^2(\theta, \vartheta),
    %\frac{r_1\rho\lambda(1-\lambda)}{2R_1}d_{\mathcal{R}}^2(\theta, \vartheta).
\end{equation*}
where $d_{\mathcal{R}}(\theta, \vartheta)<\delta.$

\subsection{The acceleration algorithm on manifolds}

In this section, we propose our acceleration algorithms for convex and non-convex functions on manifolds.
% based on the definitions and fact derived in the previous section. 
%In proposing  acceleration algorithm which we call $\mathcal{A}_2,$    
We first minimize the convex subproblem of an objective function $f$ for some existing approach $\mathcal{A}$ (such as a gradient descent algorithm) where the objective function is written as
\begin{equation*} 
    h_{*}(\vartheta)=\min_{\theta\in \mathcal{M}}\left\{f(\theta)+\frac{\kappa}{2}d_{\mathcal{R}}^2(\theta, \vartheta)\right\},
\end{equation*}
with a positive regularization parameter $\kappa$. Proposition \ref{prop1} ensures the convexity of the subproblem for an appropriate level of regularization.

Therefore, with an existing approach $\mathcal{A}$, we define the {\it proximal operator}
\begin{equation*}
    p(\vartheta)={\rm prox}_{f/\kappa}(\vartheta)=\arg\min_{\theta\in \mathcal{M}}\Big\{f(\theta)+\frac{\kappa}{2}d_{\mathcal{R}}^{2}(\theta, \vartheta)\Big\},
\end{equation*}
where $\vartheta$ is a {\it prox-center.}
%Unfortunately, there is no closed form solution to $p(\theta)$ , which therefore has to be  approximated  with  $\mathcal{A}.$ 

To consider optimizing $p(\vartheta)$, we focus on $\mathcal A$ having linear convergence rates. Specifically, a minimization algorithm $\mathcal{A},$ generating the sequence of iterates $(\theta_k)_{k\geq0},$ has a {\it linear convergence rate} if there exists $\tau_{\mathcal{A}, f}\in (0, 1)$ and a constant $C_{\mathcal{A}, f}\in\mathbb{R}$ such that
\begin{equation*}
  f(x_k)-f_{*}\leq C_{\mathcal{A}, f}(1-\tau_{\mathcal{A}, f})^k,
\end{equation*} 
where $f_{*}$ is the minimum value of $f.$

There are multiple optimization algorithms on manifolds  with linear convergence rates for strongly-convex functions on manifold. These include gradient descent, conjugate gradient descent, MASAGA \citep{Babanezhad2018MASAGAAL}, RSVRG \citep{Zhang:2016:RSF}, and many others.

For a proximal center $\vartheta$ and a smoothing parameter $\kappa$, we let
\begin{align*}
h_{\kappa}(\theta, \vartheta)=f(\theta)+\frac{\kappa}{2}d_{\mathcal{R}}^2(\theta, \vartheta).
\end{align*}

At the $k$-th iteration, given a previous iterate $\theta_{k-1}$ and the extrapolation term $\tilde{\vartheta}_{k-1},$ we perform the following steps:
\begin{enumerate}
    \item {\bf Proximal point step.}
    \begin{align*}\bar{\theta}_k\approx\arg\min_{\theta\in \mathcal{M}}h_{\kappa}(\theta, \theta_{k-1}).
    \end{align*}
    \item {\bf Accelerated proximal point step.}
    \begin{align*}\vartheta_k=\mathcal{R}_{\theta_{k-1}}\left(\alpha_k\mathcal{R}_{\theta_{k-1}}^{-1}\tilde{\vartheta}_{k-1}\right), &\qquad
        \tilde{\theta}_k\approx\arg\min_{\theta\in\mathcal{M}}h_{\kappa}(\theta, \vartheta_k), \\  \tilde{\vartheta}_k=\mathcal{R}_{\theta_{k-1}}\left(\frac{1}{\alpha_k}\mathcal{R}_{\theta_{k-1}}^{-1}\tilde{\theta}_k\right), &\qquad \frac{1-\alpha_{k+1}}{\alpha_{k+1}^2}=\frac{1}{\alpha_k^2}.
    \end{align*}
\end{enumerate}

One needs a stopping criterion, since we cannot use the functional gap  as a stopping criterion  here as in the convex case.  A stationarity stopping criterion is  adopted which consists of two conditions:
\begin{itemize}
    \item {\bf Descent condition} $h_{\kappa}(\theta, \vartheta) \leq h_{\kappa}(\vartheta, \vartheta);$
    \item {\bf Adaptive stationary condition}  ${\rm dist} \big(0_{\theta},\partial_{\theta} h_{\kappa}(\theta, \vartheta)\big) < \kappa d_{\mathcal{R}}(\theta, \vartheta).$  
    \end{itemize}
Here, ${\rm dist}(\cdot, \cdot)$ denotes the standard Euclidean distance on the tangent space.

Recall that a quadratic of the retraction distance  is added to $f$ to make the subproblem convex. So if the weak-convexity parameter $\rho$ is known, then one should set $\kappa>\rho$ to make the problem convex. In this case, it is proven that the number of inner calls to $\mathcal{A}$ for the subproblems 
\begin{equation}\label{SubProblem}
    \min_{\vartheta\in\mathcal M}h_{\kappa}(\vartheta, \theta)
\end{equation}
can be bounded by proper initialization point $\vartheta_0:$
\begin{itemize}
    \item if $f$ is smooth, then set $\vartheta_0=\theta;$
    \item if $f=f_0+\psi,$ where $f_0$ is $L$-smooth, then set $\vartheta_0={\rm prox}_{\eta\psi}\big(\mathcal{R}_{\theta}\left(\eta\nabla f_0(\theta)\right)\big)$ with $\eta\leq\frac{1}{L+\kappa}.$ 
\end{itemize}

However, in general one does not have  knowledge of $\rho.$ Thus we  propose a method that allows algorithm $\mathcal{A}_1$ (Algorithm~\ref{adapt}) to handle the convexity problem adaptively.
% inside $\mathcal{A}_2$ to handle this adaptively. 
%fix a number of iterations $T$ in advance. 

Our idea is to let $\mathcal{A}$ run on the subproblem for $T$ predefined iterations, output the point $\bar{\theta}_T,$ and check if a sufficient decrease occurs. If the subproblem is convex, then the aforementioned descent and adaptive stationary conditions are guaranteed. If either of the conditions are violated, then the subproblem is deemed non-convex. In this case, we double the value $\kappa$ and repeat the previous steps

 The tuning parameter $\kappa$ should be  chosen big enough to ensure the  convexity of the subproblems and simultaneously small enough to obtain the optimal complexity by not letting the subproblem deviate too far away from the original objective function. Thus we introduce $\kappa_{cvx}$ as an $\mathcal{A}$- {\it dependent smoothing parameter.} Notice that the linear convergence  rate $\tau_{\mathcal{A}, h_{\kappa}}$  of $\mathcal A$ is independent of the prox-center and varies with $\kappa$. We define $\kappa_{cvx}$ as
 
\begin{equation*}
    \kappa_{cvx}=\arg\max_{\kappa>0}\frac{\tau_{\mathcal{A}, h_{\kappa}}}{\sqrt{L+\kappa}}.
\end{equation*}
%Algorithm \ref{Catlalyst2}  or $\mathcal{A}_2$  relies on what we call the Adaptation algorithm  $\mathscr{Adapt}$ or Algorithm \ref{adapt} (see below equation \eqref{SubProblem}) to ensure the  smoothing parameter $\kappa$ and the number of iterations $T$  satisfy the descent and adaptive stationary condition above. 
%So for $\theta\in M,$  $\kappa$ and $T$ are chosen via Algorithm \ref{adapt}. 

\begin{algorithm}
	\caption{$\mathcal{A}_1$: The Adaptation Algorithm on Manifolds }
	\label{adapt}
	%Initialize $\theta_0=\bar{\theta};$\\
{\bf input} 
 the point $\theta\in \mathcal{M},$ the smoothing parameter $\kappa$ and the number of iterations $T$\\
	\Repeat { \rm $h_{\kappa}(\bar{\theta}_T, \theta)<h_{\kappa}(\theta, \theta)$ and ${\rm dist}(\partial h_{\kappa}(\theta_T, \theta), 0_{\theta_T})<\kappa d_{\mathcal{R}}(\theta_T, \theta)$}
	{Compute
	\begin{equation*}
    \bar{\theta}_{T}\approx\arg\min_{\vartheta\in\mathcal{M}} h_{\kappa}(\vartheta, \theta)
    \end{equation*}
    by running $T$ iterations of $\mathcal{A}$, using the initialization strategy described below Equation~\eqref{SubProblem}.\\
    {\bf If} $h_{\kappa}(\bar{\theta}_T, \theta)>h_{\kappa}(\theta, \theta)$ or ${\rm dist}(\partial h_{\kappa}(\theta_T, \theta), 0_{\theta_T})>\kappa d_{\mathcal{R}}(\theta_T, \theta)$\\
	{\bf then} go to repeat by replacing $\kappa$ with $2\kappa.$\\
	%{\bf else} go to output
	}
{\bf output} $(\theta_T, \kappa)$
\end{algorithm}

Finally, for an initial estimate $\theta_0\in \mathcal{M},$ smoothing parameters $\kappa_0, \kappa_{cvx},$  an optimization algorithm $\mathcal{A}$, and a stopping criterion based on 
%a sequence of accuracies $(\varepsilon_k)_{k\leq0},$ or  $(\delta_k)_{k\leq0},$ or 
a fixed budget $T$ and $S$, we have the following acceleration algorithm, $\mathcal A_2$, for the manifold (Algorithm \ref{Catlalyst2}).

\begin{algorithm}
	\caption{$\mathcal{A}_2$: Acceleration Algorithm on  Manifolds }
	\label{Catlalyst2}
	Initialize $\Tilde{\vartheta}_0=\theta_0,$ $\alpha=1.$\\
	\Repeat{\rm the stopping criterion is \ \ \ ${\rm dist}(\partial f(\bar{\theta}_k), 0_{\bar{\theta}})<\varepsilon$}
{for $k=1, 2, ...$
 \begin{enumerate}
\item compute $(\bar{\theta}_k, \kappa_k)=\mathcal{A}_1(\theta_{k-1}, \kappa_{k-1}, T)$
\item compute $\vartheta_k=\mathcal{R}_{\theta_{k-1}}\left(\alpha_k\mathcal{R}^{-1}_{\theta_{k-1}}\Tilde{\vartheta}_{k-1}\right)$ and apply 
      $S_k\log(k+1)$ iterations of $\mathcal{A}_1$ to find 
   \begin{equation*}
    \Tilde{\theta}_k\approx\arg\min_{\theta\in\mathcal{M}} h_{\kappa_{cvx}}(\theta, \vartheta_k),
   \end{equation*}
   by using initialization strategy described below \eqref{SubProblem}.
\item Update $\Tilde{\vartheta}_k$ and  $\alpha_{k+1}$:
 \begin{align*}
    \tilde{\vartheta}_k&=\mathcal{R}_{\theta_{k-1}}\left(\frac{1}{\alpha_{k}}\mathcal{R}^{-1}_{\theta_{k-1}}\Tilde{\theta}_k\right),\\
    \alpha_{k+1}&=\frac{\sqrt{\alpha_k^4+4\alpha_k^2}-\alpha_k^2}{2}.
 \end{align*}
\item Choose $\theta_k$ to be any point satisfying $f(\theta_k)=\min\{f(\bar{\theta}_k), f(\Tilde{\theta}_k)\}.$
\end{enumerate}
}
\end{algorithm}

\begin{remark}
Note that there are two sequences  $\left\{\tilde{\theta}_k\right\}$  and $\left\{\bar{\theta}_k\right\}$ in Algorithm $\mathcal A_2$. Since the extrapolation step is designed for the convex case, the second sequence $\{\tilde{\theta}_k\}$ approximates the optimal point with accelerated rate which means that it approaches the optimal point faster than the first sequence $\left\{\bar{\theta}_k\right\}$ above. Intuitively, when the first sequence is chosen it uses the initial algorithm $\mathcal A$ and adapts the smoothing parameter to our objective--implying that the Nesterov step failed to accelerate convergence.
%Intuitively, when the first sequence is chosen, which  simply uses the initial algorithm $\mathcal{A}$ and also adapts the smoothing parameter to our objective, it implies that Nesterov's step fails to accelerate convergence.
\end{remark}

% The tuning $\kappa$ should be  chosen big enough to ensure the strong convexity of the subproblems and simultaneously small enough to obtain the optimal complexity. Thus we introduce $\kappa_{cvx}$ as $\mathcal{A}$ {\it dependent smoothing parameter.} Notice that the linear convergence  rate $\tau_{\mathcal{A}, h_{\kappa}}$  of $\mathcal A$ is independent of the prox-center and varies with $\kappa,$ so define $\kappa_{cvx}$ as
%\begin{equation*}
%    \kappa_{cvx}=\arg\max_{\kappa>0}\frac{\tau_{\mathcal{A}, h_{\kappa}}}{\sqrt{L+\kappa}}.
%\end{equation*}

In the adaptation method $\mathcal{A}_1(\theta_{k-1}, \kappa_{k-1}, T)$, the resulting $\bar{\theta}_k$ and $\kappa_k$ have to satisfy the following inequalities
\begin{align}\label{C1}
    {\rm dist}\big(0_{\bar{\theta}_k}, \partial h(\bar{\theta}_k, \theta_{k-1})\big)&<\kappa_k d_{\mathcal{R}}(\bar{\theta}_k, \theta_{k-1}) \quad \text{and} \\
    \quad h_{\kappa_k}(\bar{\theta}_k, \theta_{k-1})&\leq h_{\kappa_k}(\theta_{k-1}, \theta_{k-1}).
\end{align}
The resulting $\Tilde{\theta}_k,$ %with a similar predefined iteration strategy, 
needs to satisfy the condition that if the function $f$ is convex, then
\begin{equation}
    \label{C2}
    {\rm dist}\big(0_{\Tilde{\theta}_k}, \partial h_{\kappa_{cvx}}(\Tilde{\theta}_k, \vartheta_k)\big)<\frac{\kappa_{cvx}}{k+1}d_{\mathcal{R}}(\Tilde{\theta}_k, \vartheta_k).
\end{equation}
We then have the following lemma:
%\begin{lemma}
%Suppose the sequence $\{\alpha_k\}_{k\geq1}$ is produced by $\mathcal{A}_2.$ Then, the following bounds hold for all $k\geq1$
%\begin{equation*}
%    \frac{\sqrt{2}}{k+2}\leq\alpha_k\leq\frac{2}{k+1}.
%\end{equation*}
%\end{lemma}
%
\begin{lemma}
\label{lem-2}
Suppose $\theta$ satisfies ${\rm dist}(0_{\theta}, \partial h_{\kappa}(\theta, \vartheta))<\varepsilon,$ and $|\nabla d_{\mathcal{R}}^{2}(\theta, \vartheta)|\leq Kd_{\mathcal{R}}(\theta, \vartheta)$, then the inequality holds:
\begin{equation*}
    {\rm dist}(0_{\theta}, \partial f(\theta))\leq\varepsilon+\kappa Kd_{\mathcal{R}}(\theta, \vartheta).
\end{equation*}
\end{lemma}

\begin{proof}
We can find $v\in\partial h_{\kappa}(\theta, \vartheta)$ with $\|v\|\leq\varepsilon.$ Taking into account $\partial h_{\kappa}(\theta, \vartheta)=\partial f(\theta)+\kappa\nabla d_
{\mathcal{R}}^{2}(\theta, \vartheta)$ the result follows.
\end{proof}
Since we assume retraction distance function $d_{\mathcal{R}}$ is continuous, we can deduce that the vector field $\nabla d_{\mathcal{R}}^{2}(\theta, \vartheta)$ is continuous, so the conditions of Lemma \ref{lem-2} are very mild. Also, as mentioned previously, the square retraction distance function $d_{\mathcal{R}}^2(\cdot, \vartheta)$ acts like a square normal function in a small neighborhood of $\vartheta$.

%We have the following generalization of the geodesically strongly convex functions.
We define the following retraction-based strongly convex function: 
\begin{definition}
A function $f$ is {\it $\mu$-strongly convex with respect to the retraction} $\mathcal{R}$ if for any points $\theta_1, \theta_2\in\mathcal{M}$ and $\mu>0$ the inequality holds
\begin{equation}
\label{f-strongly-convex-2}
    f(\theta_2)\geq f(\theta_1)+\langle v, \mathcal{R}_{\theta_1}^{-1}\theta_2\rangle+\frac{\mu}{2}d_{\mathcal{R}}^2(\theta_1, \theta_2), \;  v\in\partial f(\theta_1).
\end{equation}
\end{definition}

Then we have the following convergence analysis for the acceleration algorithm $\mathcal{A}_2$:

\begin{theorem}
Fix real-valued constants $\kappa_0, \kappa_{cvx}>0$ and the point $\theta_0\in\mathcal{M}.$ Set $\kappa_{\max}=\max_{k\geq1}\kappa_k.$ Suppose that the number of iterations  $T$  is such that $\bar{\theta}_k$ satisfies \eqref{C1}, and $|\nabla d_{\mathcal{R}}^{2}(\theta, \vartheta)|\leq Kd_{\mathcal{R}}(\theta, \vartheta).$ Define $f^{*}=\lim_{k\rightarrow\infty}f(\theta_k).$ Then for any $N\geq1,$ the iterated sequence generated by the 
acceleration algorithm satisfies
\begin{equation*}
        \min_{j=1, ..., N}\Big\{{\rm dist}^2\big(0_{\bar{\theta}_j}, \partial f(\bar{\theta}_j)\big)\Big\}\leq\frac{8\kappa_{\max}K^2}{N}\left(f(\theta_0)-f^{*}\right).
\end{equation*}
If in addition the function $f$ is $\kappa_{cvx}(K_1^4K_2^4-R_1)$-strongly convex and $S_k$ is chosen so that $\Tilde{\theta}_k$ satisfies \eqref{C2}, then
\begin{equation}\label{acceleration}
   f(\theta_N)-f^{*}\leq\frac{4\kappa_{cvx} K_1^2K_2^2}{(N+1)^2}d_{\mathcal{R}}^2(\theta^{*}, \theta_0),
\end{equation}
  where $\theta^{*}$ is any minimizer of the function $f.$
\end{theorem}

The detailed proof of this theorem can be found in the Appendix.

\begin{remark}
If the original method $\mathcal A$ has a linear rate of convergence then our method $\mathcal A_2$ also converges to the local minimum for the strongly convex case.  If the knowledge of the strong-convexity is given, then some existing method can achieve optimal linear rate for smooth and convex functions \citep{pmlr-v49-zhang16b}, however, it is overall an extremely difficult to verify convexity of a function on a manifold, and our method adapts to that without requiring the knowledge of the convexity. 
Note that our algorithm also applies to the subgradient descent method, where instead of gradient of the function one takes the subdifferential, for non-smooth functions. In this case, for the strongly-convex objective, the subgradient method converges to the optimum with  $O(1/N)$  rate of convergence (see \cite{pmlr-v49-zhang16b}). Thus our  accelerated rate  $O(1/N^2)$ can be considered optimal for strongly-convex functions on the manifold.
%\lizhen{[Delete this statement:] From the last equation \eqref{acceleration} one can see that the proposed algorithm $\mathcal A_2$ has the optimal rate $O\Big(\sqrt{\frac{1}{\varepsilon}}\Big)$ of convergence in the strongly-convex case.}
%, which is a nice improvement for the gradient based methods that have rate of convergence $O\Big(\frac{1}{\varepsilon}\Big).$  
%From the last equation \eqref{acceleration} one can see that the proposed algorithm $\mathcal A_2$ has the accelerated rate $O\Big(\sqrt{\frac{1}{\varepsilon}}\Big)$ of convergence in the strongly-convex case, whereas the initial algorithm $\mathcal A$ has a linear rate of convergence $O\Big(\frac{1}{\varepsilon}\Big).$
\end{remark}

%\begin{remark}
%Our method might not be optimal for the convex case due to not having the knowledge of the convexity or non-convexity of the objective function. We wish to mention that it is overall a very difficult  task to verify that an objective function over manifolds is convex, strongly convex or non-convex, especially for complex objective
% functions over manifolds. Our key philosophy is that our approach can adapt to the weak convexity of the objective in the sense that if
% the function is indeed strongly convex, we gain acceleration, while guaranteeing convergence to stationary point if the
% objective function is weakly convex (not-convex).
%\end{remark} 
%

%\begin{remark}
%Since the manifold doesn't have global coordinates and, moreover, the inverse retraction mapping is defined locally, the convergence to the minimum is guaranteed only locally.
%\end{remark}

\section{Simulation study and data analysis}\label{sec-simu}
To examine the convergence and acceleration rates of our proposed algorithm, we first apply our method to the estimation of both intrinsic and extrinsic Fr\'echet means on spheres, in which one has the exact optima for comparison in the case of extrinsic mean. We also apply our algorithm to the Netflix movie-ranking data set as an example of optimization over Grassmannian manifolds in the low-rank matrix completion problem. 

%In the following results, we demonstrate the utility of our algorithm both for high dimensional manifold-valued data (Section~\ref{subsec:frechet}) and Euclidean space data with non-Euclidean parameters (Section~\ref{subsec:netflix}). We wrote the code for our implementations in Python and carried out the parallelization of the code through MPI\footnote{Our code is available at \url{https://github.com/michaelzhang01/parallel_manifold_opt}}\cite{Dalcin:2005}.

\subsection{Estimation of intrinsic Fr\'echet means on manifolds}\label{subsec:frechet}
We first consider the estimation problem of Fr\'echet means on manifolds \citep{frechet}. % In particular, the manifold under consideration is the sphere in which we wish to estimate the  mean of a sphere \cite{linclt}. 
In this simple example, we have observations $\{x_1,\ldots, x_N\}$ that lie on a sphere $ \mathbb{S}^{d} $ and our goal is to estimate the sample mean:
\begin{align}
\hat\theta=\arg\min_{\theta\in S^d} f(\theta), \hspace{1em} f(\theta) = \sum_{i=1}^n \rho^2(\theta, x_i).
\end{align}
If $\rho$ is the embedded distance metric in the Euclidean space, then there exists a closed form solution $\hat\theta=\sum_{i=1}^N x_i/\| \sum_{i=1}^N x_i\|$, which is the projection of the Euclidean mean $\bar{x}$ onto the sphere \citep{rabibook}.  This is called the \emph{extrinsic mean.} When $\rho$ is taken to be the geodesic or intrinsic distance, $\hat\theta$ is called the \emph{intrinsic mean}.  We will consider estimation of both extrinsic and intrinsic means using our method compared to other optimization techniques. 

One simple examples of a retraction map for $\mathbb{S}^d$ is
\begin{equation*}
    \mathcal{R}_{\vartheta}v=\frac{\vartheta+v}{|\vartheta+v|},
\end{equation*}
where $|\cdot|$ is the Euclidean norm in $\mathbb{R}^{d+1}.$
Therefore the inverse retraction has the following expression
\begin{equation*}
    \mathcal{R}_{\vartheta}^{-1}\theta=\frac{1}{\vartheta^T\theta}\theta-\vartheta.
\end{equation*}

%However, if $ \rho $ is the arc length of points measured on the great circle, $\rho(\theta,x) = \arccos( \langle \theta,x \rangle)$, then we no longer have a closed form expression for the sample mean and therefore we must estimate the intrinsic mean by optimizing the loss function computationally.

We first compare our accelerated method against gradient descent optimization and a Newton-type optimization scheme, DANE \citep{Shamir}, and a Nesterov method, RAGD \citep{zhang2018towards}, adapted for manifolds. For all the experiments in this section,  we optimized the step size of the optimizer  using an Armijo condition backtracing line seach \citep{armijo1966minimization} where we reduce the step size by a factor of $ .95 $ until the difference between the old loss function evaluation and the new one is $ 10^{-5} \times .95$. For our Catalyst algorithm manifold we set the A2 budget to $ S=10 $, the A1 number of iterations to $ T=5 $, and cutoff parameter for A1 is initialized at $ .1 $. For the DANE results, we set the regularization term to $ 1 $, For RAGD we set the shrinkage parameter to $ 1 $. Our synthetic data set is 10,000 observations generated i.i.d from a $100$ dimensional $\mbox{N}(0, I)$ distribution projected onto $ \mathbb{S}^{99} $. 

We run each optimization routine for 100 iterations. Figure~\ref{fig:intrinsic_comparison} and Figure~\ref{fig:extrinsic_comparison} shows that our novel accelerated method converges, for an intrinsic mean as well as an extrinsic mean example, to an optima in fewer iterations than the other competing methods, both in terms of the loss function value and the norm of the loss function gradient. Moreover, we can see in the intrinsic mean example, our method is able to obtain a smaller loss function and gradient norm than the competing methods. In the extrinsic mean example, our method obtains a comparable loss function value and MSE between the learned parameter and the closed-form expression of the sample mean with other methods in fewer iterations and obtains a smaller gradient norm than the competing methods.

By explicit calculation we show the objective functions are strongly convex over a neighborhood of any point on the manifold (see the Appendixx for a proof). This is a highly non-trivial task for general objective functions, hence necessitating an adaptive method such as ours. Moreover, in the extrinsic mean example, since we have a closed form expression of the Fr\'echet mean we also show that our optimization approach converges to the true extrinsic mean in terms of mean squared error faster than the other optimization methods.  %Our simulation study demonstrates  accelerated convergence.
\begin{figure}
	\centering
	\includegraphics[width=.5\linewidth]{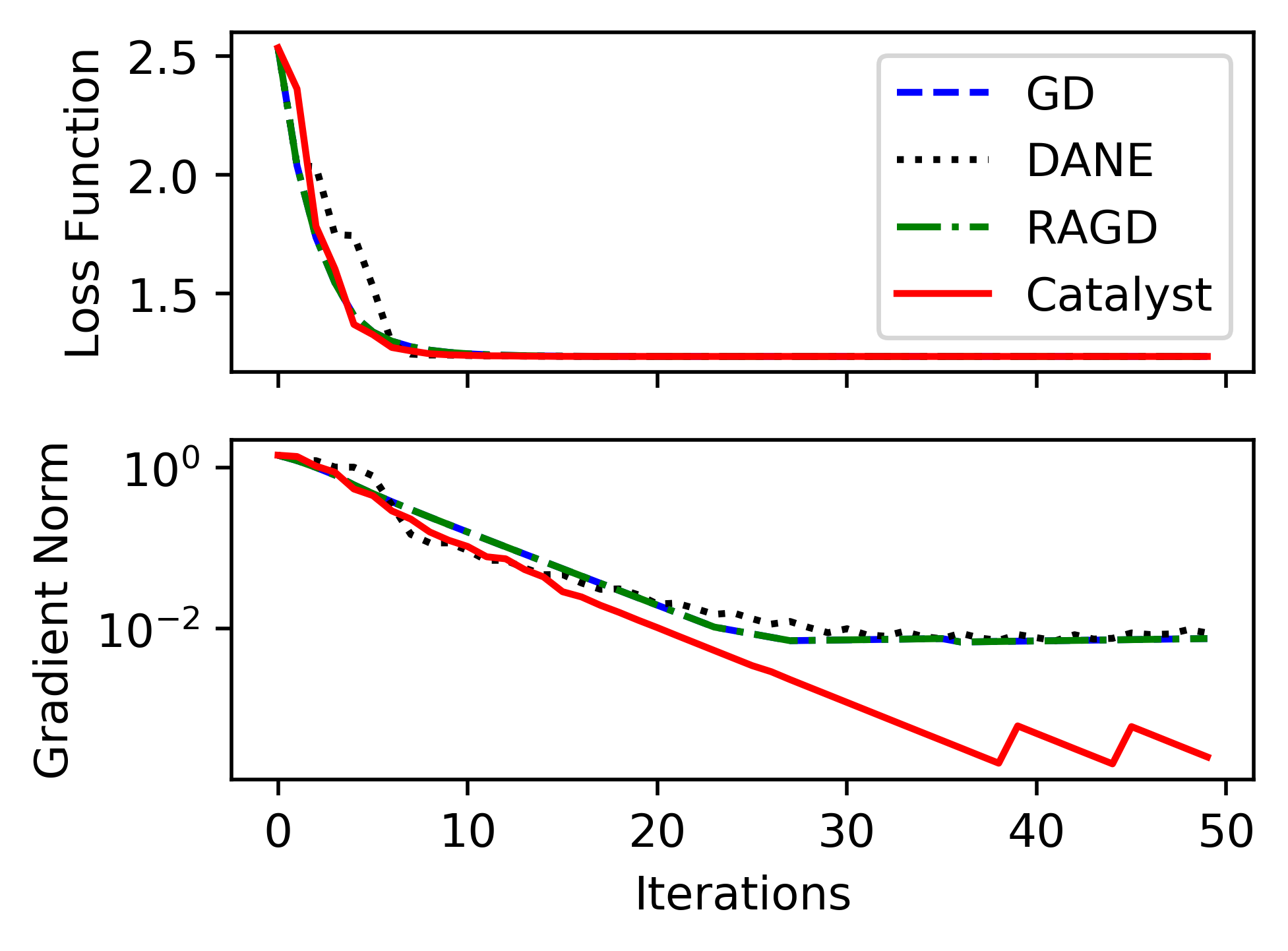}\includegraphics[width=.5\linewidth]{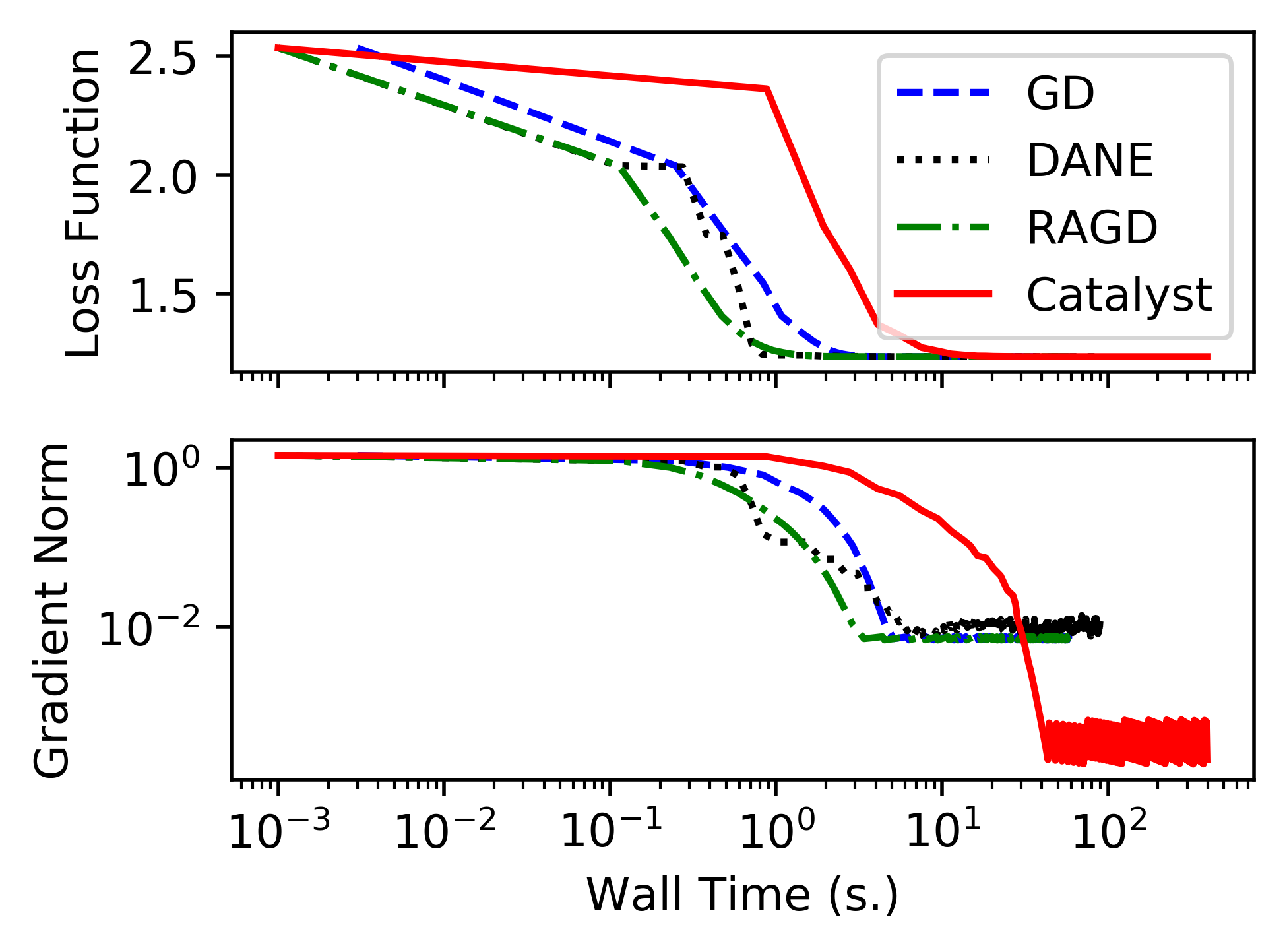}
	\caption{Intrinsic mean comparison on spheres} 
	\label{fig:intrinsic_comparison}
\end{figure}
\begin{figure}
	\centering
	\includegraphics[width=.5\linewidth]{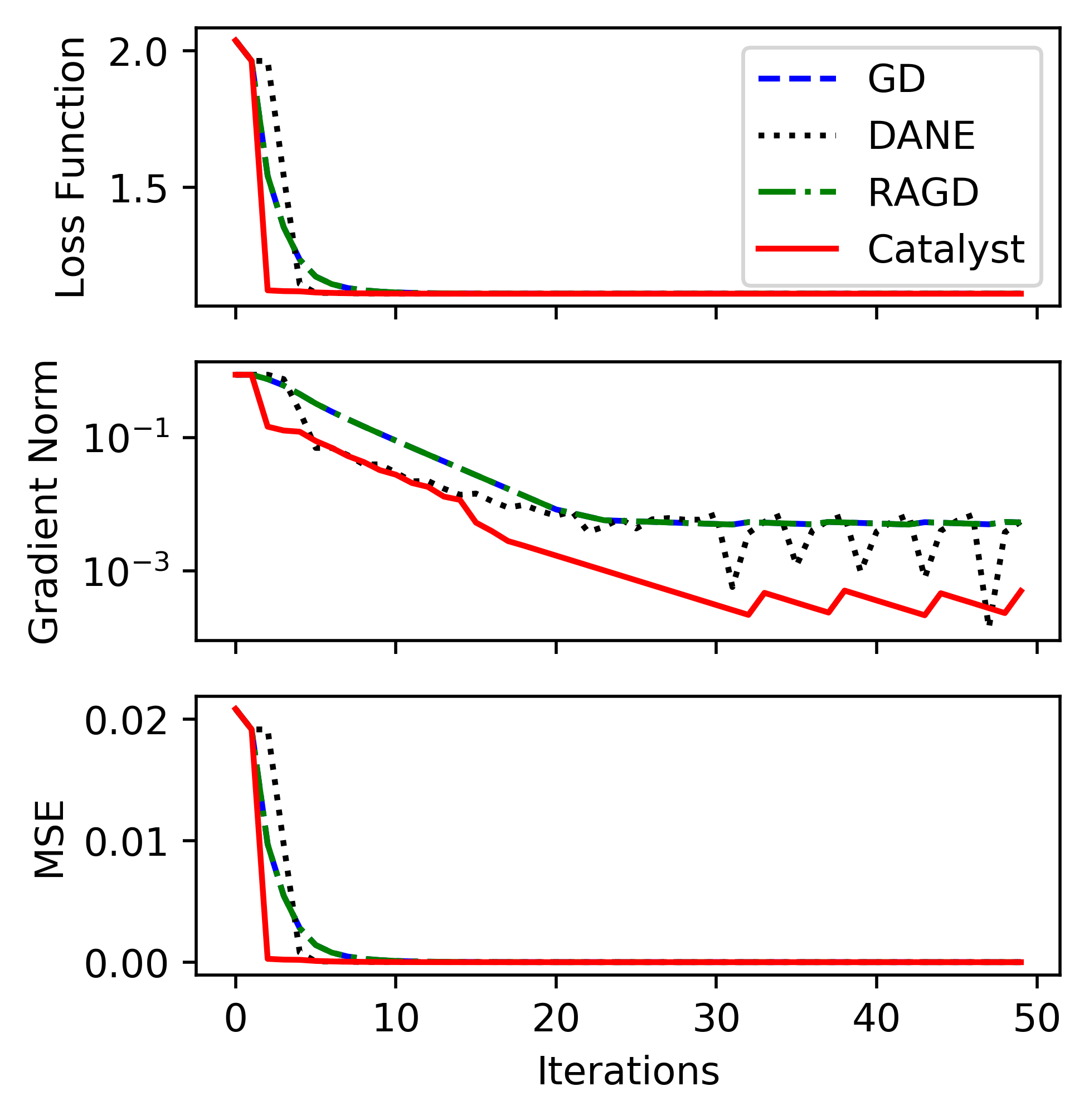}\includegraphics[width=.5\linewidth]{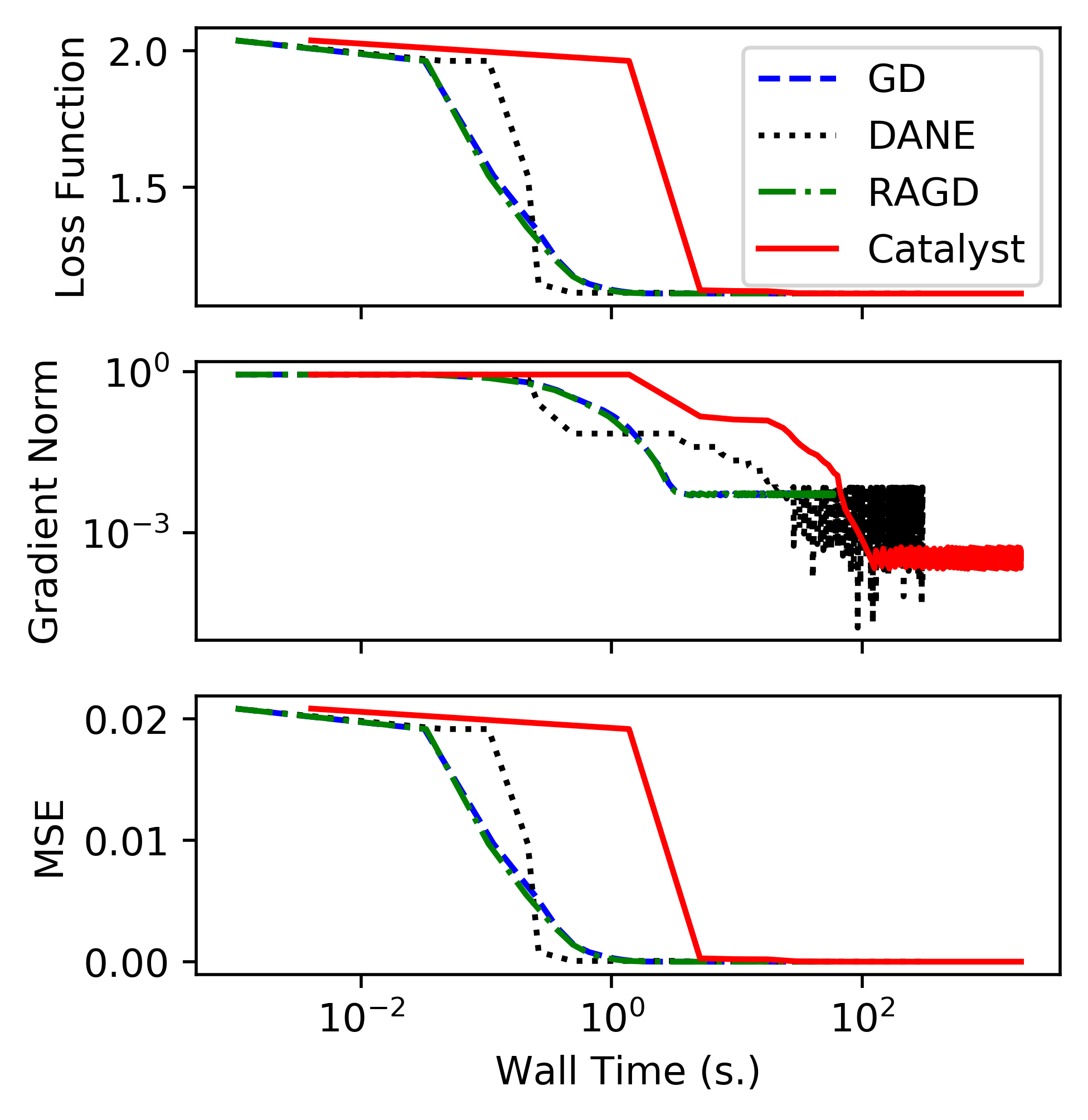}
	\caption{Extrinsic mean comparison on spheres} 
	\label{fig:extrinsic_comparison}
\end{figure}
\subsection{Real data analysis: the Netflix example}\label{subsec:netflix}
Next, we consider an application of our algorithm to the Netflix movie rating dataset. This dataset of over a million entries, $ X \in \mathbb{R}^{M \times N} ,$ consists of $M = 17770$ movies and $N = 480189$ users, in which only a sparse subset of the users and movies have ratings. In order to build a better recommendation systems to users, we can frame the problem of predicting users' ratings for movies as a low-rank matrix completion problem by learning the rank-$ r $ Grassmannian manifold $ U \in \mbox{Gr(M, r)} $ which optimizes for the set of observed entries $ (i,j) \in \Omega$ the loss function
 \begin{equation}
		L(U) = \frac{1}{2} \sum_{(i,j) \in \Omega} \left\{ (UW)_{ij} - X_{ij} \right\}^{2} + \frac{\lambda^{2}}{2} \sum_{(i,j) \notin \Omega}(UW)_{ij},		
\end{equation}
where $W$ is an $r$-by-$N$ matrix. Each user $ k$ has the loss function $\mathcal{L}(U,k)=\frac{1}{2} \left|  c_k\circ\left(  Uw_k(U)-X_k \right)     \right|^{2}$ , where $\circ$ is  the Hadamard product, $(w_k)^{i}=W_{ik},$ and
\begin{equation*}
\begin{split}
    (c_{k})^i=\begin{cases}
    1, & {\rm if} \ \ \ (i, k)\in \Omega \\
    \lambda, & {\rm if} \ \ \ (i, k)\notin \Omega 
    \end{cases},
    \qquad
    (X_{k})^i=\begin{cases}
    X_{ik}, & {\rm if} \ \ \ (i, k)\in \Omega \\
    0, & {\rm if} \ \ \ (i, k)\notin \Omega, 
    \end{cases}\\
    w_k(U)=\big(U^T{\rm diag}(c_k\circ c_k)U\big)^{-1}U^T\big(c_k\circ c_k\circ X_k\big).
\end{split}
\end{equation*}
This results in the following gradient
\begin{align*}
    \nabla\mathcal{L}(U, k)&= \big(c_k\circ c_k\circ(Uw_k(U)-X_k)\big)w_k(U)^T\\
    &= {\rm diag}(c_k\circ c_k)(Uw_k(U)-X_k)w_k(U)^T.
\end{align*}

For this problem on Grassman manifolds, we have the retraction map:

\begin{equation}
\mathcal{R}_{V}U = U+V
\end{equation}

and the inverse retraction map:

\begin{equation}
\mathcal{R}_{V}^{-1}U = V - U(U^{T}U)^{-1}U^{T}V
\end{equation}

%Due to the scale of the data, 
We look at a comparison of our method against a standard gradient descent method on a subset of the data where we only observe a million ratings ($ \approx 1.5\%$ of the full data set). In this setting we fix the matrix rank $ r=5 $ and the regularization parameter $ \lambda = .01 $. Figure~\ref{fig:parallel_netflix} shows that our accelerated method obtains a smaller loss function value, a smaller identical test set MSE, and nearly identical loss gradient norm faster than RAGD, DANE, or a typical gradient descent approach.

On a large scale, we apply a parallelized version of our accelerated method  and a communication-efficient parallel algorithm on manifolds proposed in \cite[ILEA]{lizhennips2018} on the full Netflix dataset. We randomly distribute the data across 64 processors and run the optimization routine for 200 iterations. In Figure~\ref{fig:parallel_netflix}, again we can see steady acceleration that our method provides in terms of the loss function value across iterations and the loss of gradient norm though ILEA obtains slightly better test set MSE than our method. 
%However, our algorithm achieves more rapid convergence. In particular, we plot the loss function  (in the third panel of the plot) as a function of iterations from 50 to 200 in which one can see the steady improvement of our algorithms. 
% noting the scale of the y-axis, the difference of their predictive performances is \lizhen{quite small}.
\begin{figure}
	\centering
	\includegraphics[width=.5\linewidth]{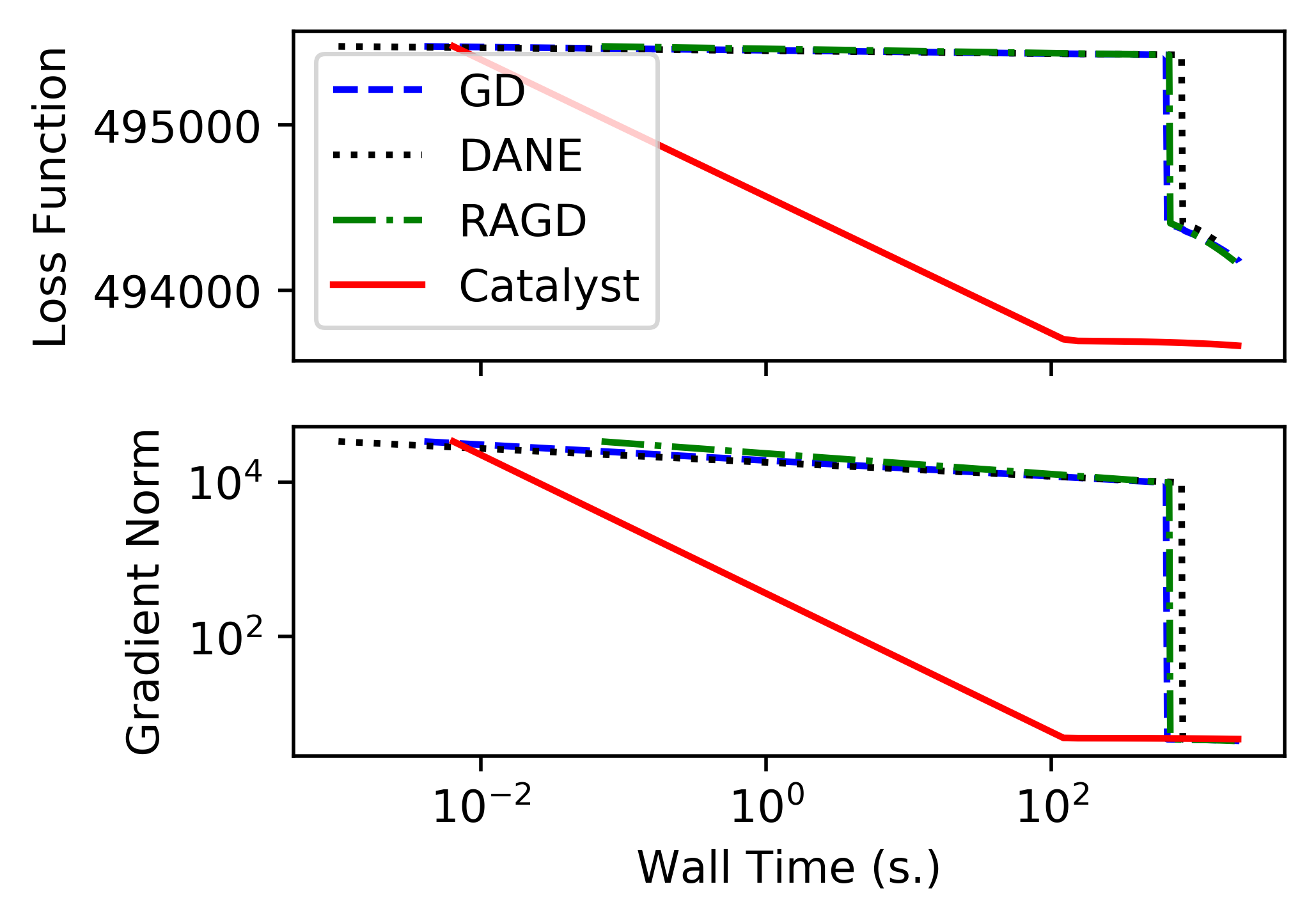}\includegraphics[width=.5\linewidth]{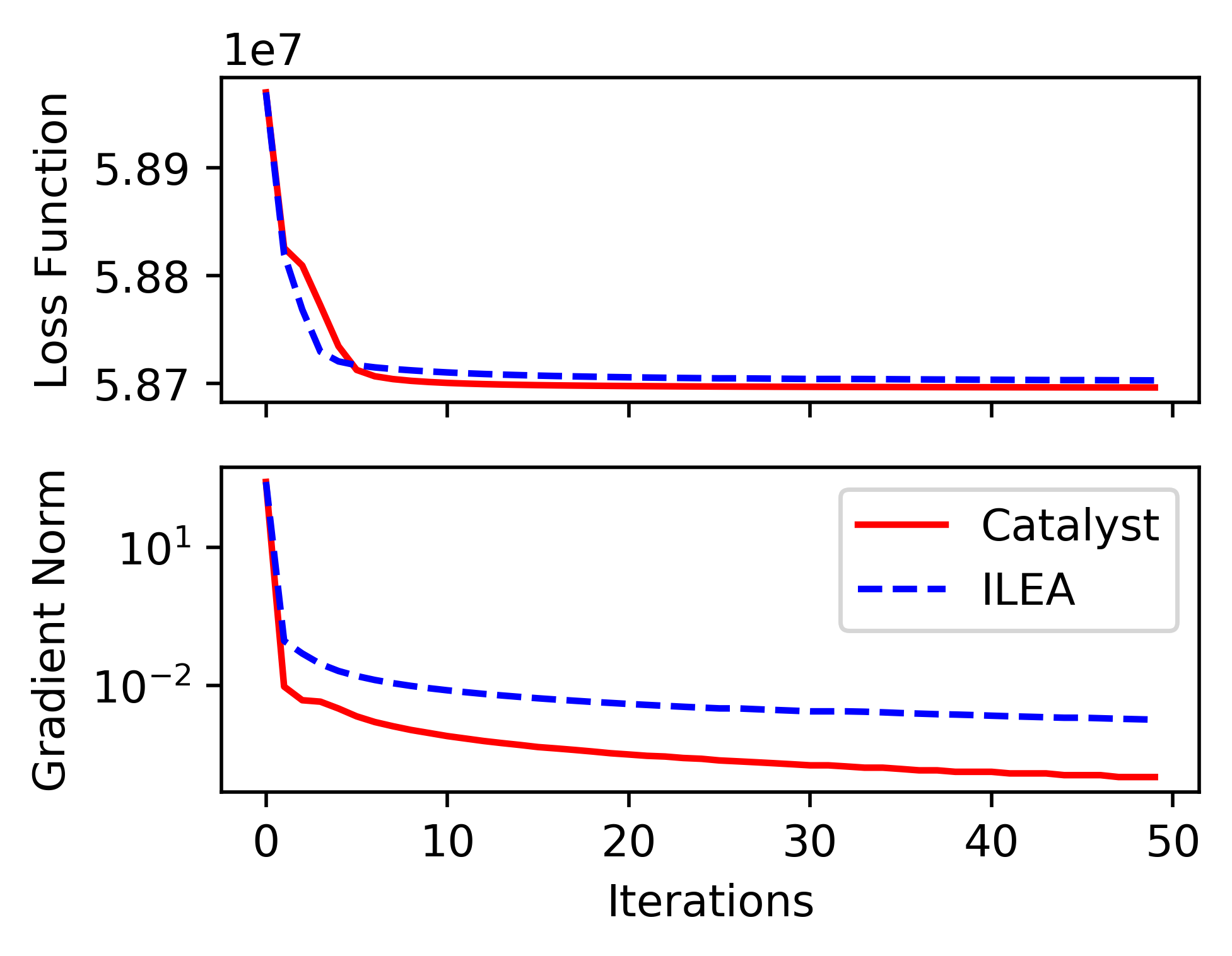}\\
	\includegraphics[width=.5\linewidth]{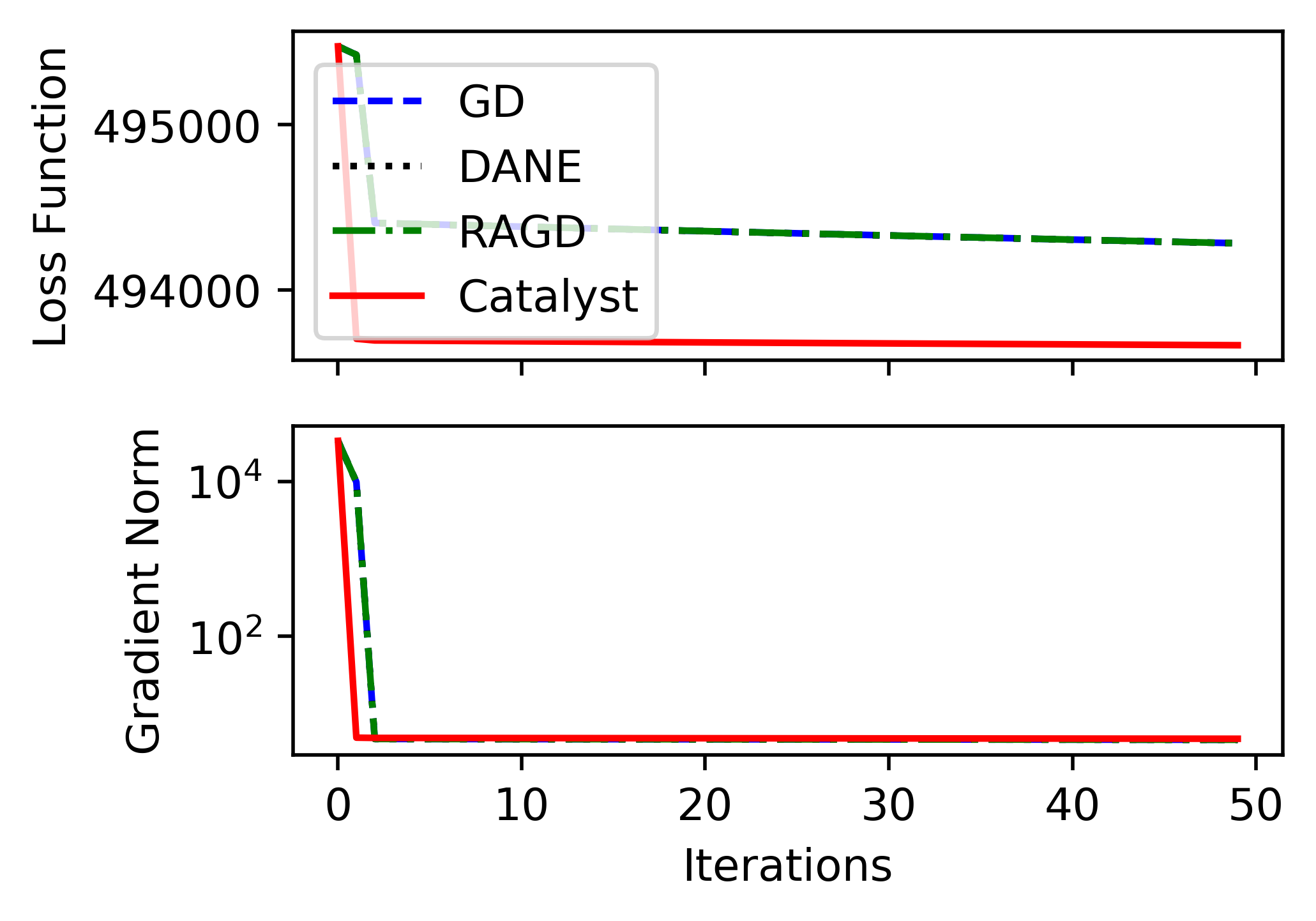}\includegraphics[width=.5\linewidth]{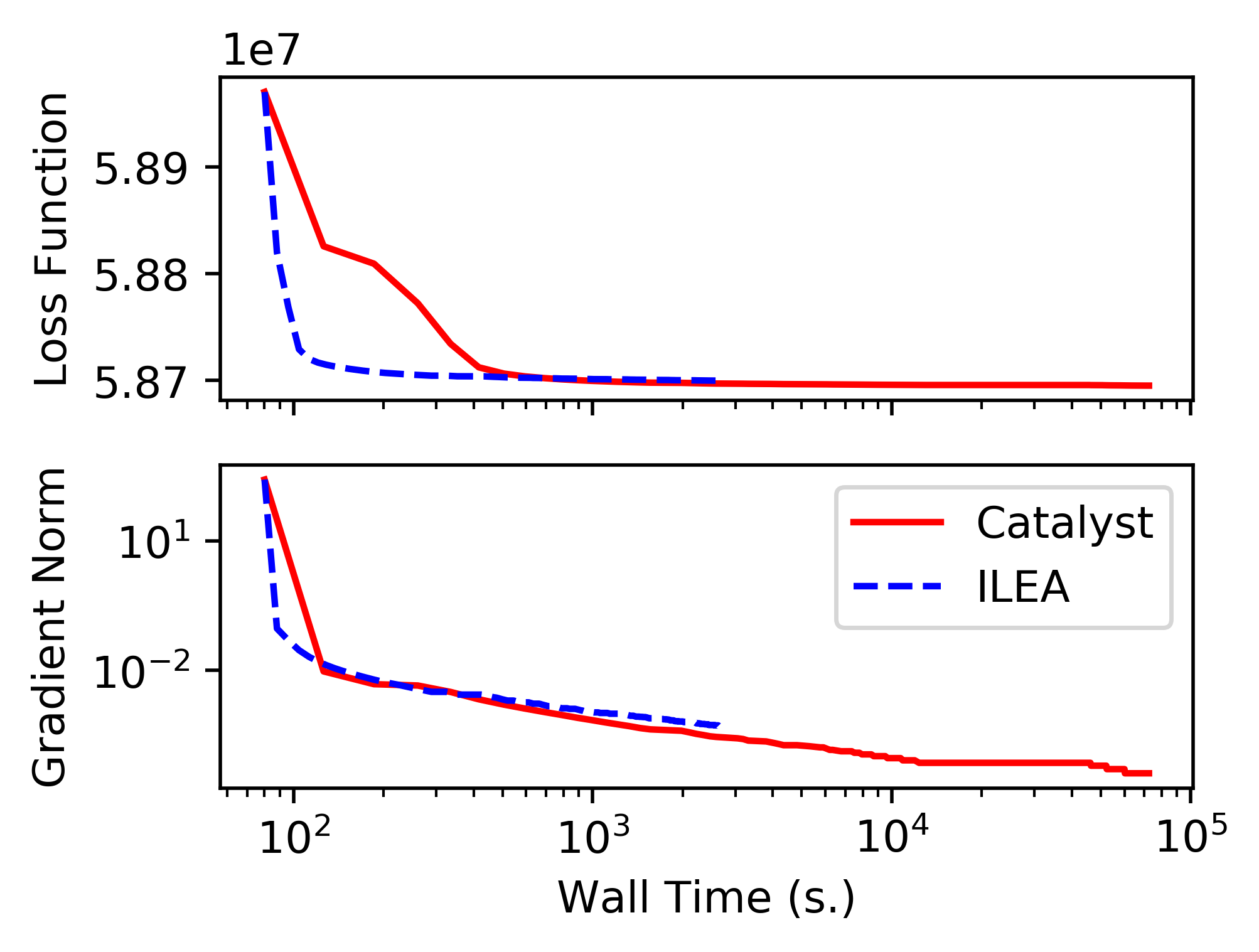}
%	\caption{Acceleration results for the Netflix example.}\label{fig:netflix}.
%\end{figure}
%\begin{figure}
%	\centering
	\caption{Results for the parallel (right) and reduced (left) Netflix example.}\label{fig:parallel_netflix}.
\end{figure}

\section{Conclusion and Discussion }\label{sec:conclusion}

We propose a general scheme for solving non-convex optimization on manifolds which yields theoretical guarantees of convergence to a stationary point when the objective function is non-convex. When the objective function is convex, it leads to accelerated  convergence rates for a large class of first order methods, which we show in our numerical examples. One of the interesting future directions  we want to pursue is  proposing accelerated algorithms on statistical manifolds (manifolds of densities or distributions) by employing information-geometric techniques, and applying the algorithms to accelerate convergence and mixing MCMC algorithms. %that can potentially lead to better mixing rates. 

%We propose in this paper a communication efficient parallel algorithm for general optimization problems on manifolds which is applicable to many different manifold spaces and loss functions. Moreover, our proposed algorithm can explore the geometry of the underlying space efficiently and perform well in simulation studies and practical examples all while having theoretical convergence guarantees.
%
%In the age of ``big data'', the need for distributable inference algorithms is crucial as we cannot reliably expect entire datasets to sit on a single processor anymore. Despite this, much of the previous work in parallel inference has only focused on data and parameters in Euclidean space. Realistically, much of the data that we are interested in is better modeled by manifolds and thus we need fast inference algorithms that are provably suitable for situations beyond the Euclidean setting. In future work, we aim to extend the situations under which parallel inference algorithms are generalizable to manifolds and demonstrate more critical problems (in neuroscience or computer vision, for example) in which parallel inference is a crucial solution.

\section{Appendix}
\subsection{Proof to Theorem 2}
\label{suppl-proofs}

We first introduce a simple lemma.

%We then have the following lemmas.
\begin{lemma}
Suppose the sequence $\{\alpha_k\}_{k\geq1}$ is produced by $\mathcal{A}_2.$ Then, the following bounds hold for all $k\geq1$
\begin{equation*}
    \frac{\sqrt{2}}{k+2}\leq\alpha_k\leq\frac{2}{k+1}.
\end{equation*}
\end{lemma}

%To better address the challenges arising from these applications.  Many of the data of 
%   there is a critical need for developing valid models and efficient inference methods  for statistical inference related to such data.

%5 most data in machine vision and pattern recognition is intrinsically non-Euclidean, i.e. standard Euclidean calculus does not apply.

\begin{proof}[Proof of Theorem 2]
The descent condition in 
\begin{align}\label{C1}
\begin{split}
    {\rm dist}\big(0_{\bar{\theta}_k}, \partial h(\bar{\theta}_k, \theta_{k-1})\big)&<\kappa_k d_{\mathcal{R}}(\bar{\theta}_k, \theta_{k-1}) \quad \text{and} \\
    \quad h_{\kappa_k}(\bar{\theta}_k, \theta_{k-1})&\leq h_{\kappa_k}(\theta_{k-1}, \theta_{k-1}),
    \end{split}
\end{align} 
implies $\{f(\theta_k)\}_{k\geq0}$ are monotonically decreasing. 
From this 
\begin{equation}\label{T-1}
\begin{split}
    f(\theta_{k-1})&=h_{\kappa}(\theta_{k-1}, \theta_{k-1})\\
                          &\geq h_{\kappa}(\bar{\theta}_k, \theta_{k-1})\\
                          &\geq f(\theta_k)+\frac{\kappa}{2}d_{\mathcal{R}}^2(\bar{\theta}_k, \theta_{k-1}).
\end{split}
\end{equation}
Using condition \eqref{C1}, we apply Lemma 2 with $\vartheta=\theta_{k-1}, \theta=\bar{\theta}_k$ and $\varepsilon=\kappa Kd_{\mathcal{R}}(\bar{\theta}_k, \theta_{k-1});$ hence
\begin{equation*}
    {\rm dist}(0_{\bar{\theta}_k}, \partial f(\bar{\theta}_k))\leq 2\kappa K d_{\mathcal{R}}(\bar{\theta}_k, \theta_{k-1}).
\end{equation*}
Combining the above inequality with \eqref{T-1}, one has
\begin{align}\label{T-2}
\begin{split}
    {\rm dist}^2(0_{\bar{\theta}_k}, \partial f(\bar{\theta}_k))&\leq4\kappa^2K^2d_{\mathcal{R}}^2(\bar{\theta}_k, \theta_{k-1})\\
    &\leq8\kappa_{\max}K^2\big(f(\theta_{k-1}-f(\theta_k)\big).
\end{split}
\end{align}
Summing $j=1$ to $N,$ we can conclude
\begin{align*}
        \min_{j=1, ..., N}\Big\{{\rm dist}^2\big(0_{\bar{\theta}_j}, \partial f(\bar{\theta}_j)\big)\Big\}&\leq\frac{8\kappa_{\max}K^2}{N}\sum^N_{j=1}\big(f(\theta_{j-1})-f(\theta_j)\big)\\
        &\leq\frac{8\kappa_{\max}K^2}{N}\big(f(\theta_0)-f^{*}\big).
\end{align*}

Fix an $v_k\in\partial h_{\kappa}(\Tilde{\theta}_k, \vartheta_k).$ Since the function $f$ is $\kappa_{cvx}(K_1^4K_2^4-R_1)$-strongly convex, the function $h_{\kappa_{cvx}}$ is $\kappa_{cvx} K_1^4K_2^4$-strongly convex.
\begin{align*}
    f(\theta)&+\frac{\kappa_{cvx}}{2}d_{\mathcal{R}}^2(\theta, \vartheta_k)\\
   & \geq f(\Tilde{\theta}_k)+\frac{\kappa_{cvx}}{2}d_{\mathcal{R}}^2(\Tilde{\theta}_k, \vartheta_k)\\
  & \quad\quad+\frac{\kappa_{cvx} K_1^4K_2^4}{2}d_{\mathcal{R}}^2(\Tilde{\theta}_k, \theta)+\langle v_k, \mathcal{R}_{\Tilde{\theta}_k}^{-1}\theta\rangle.
\end{align*}
Then 
\begin{align*}
    f(\Tilde{\theta}_k) &\leq 
    f(\theta)+\frac{\kappa_{cvx}}{2}\big(d_{\mathcal{R}}^2(\theta, \vartheta_k)-K_1^4K_2^4 d_{\mathcal{R}}^2(\Tilde{\theta}_k, \theta)\\
   &\quad\quad -d_{\mathcal{R}}^2(\Tilde{\theta}_k, \vartheta_k)\big)-\langle v_k, \mathcal{R}_{\Tilde{\theta}_k}^{-1}\theta\rangle.
\end{align*}
So for any $\theta\in\mathcal{M}$
\begin{equation*}
\begin{split} 
    f(\theta_k)\leq& f(\Tilde{\theta}_k)\\
    \leq&  f(\theta)+\frac{\kappa_{cvx}}{2}\Big(K_1^2\|\mathcal{R}^{-1}_{\theta_{k-1}}\theta-\mathcal{R}^{-1}_{\theta_{k-1}}\vartheta_{k}\|^2- K_1^4K_2^2\|\mathcal{R}^{-1}_{\theta_{k-1}}\Tilde{\theta}_k\\
    &-\mathcal{R}^{-1}_{\theta_{k-1}}\theta\|^2
    \Big)-\frac{\kappa_{cvx}}{2}d_{\mathcal{R}}^2(\Tilde{\theta}_k, \vartheta_k)-\langle v_k, \mathcal{R}_{\Tilde{\theta}_k}^{-1}\theta\rangle.
   \end{split}
\end{equation*}
We substitute $\theta=\mathcal{R}_{\theta_{k-1}}\alpha_k\mathcal{R}^{-1}_{\theta_{k-1}}\theta^{*},$ where $\theta^{*}$ is any minimizer of $f.$ Using convexity of $f$
\begin{equation*}
    f(x)\leq\alpha_k f(\theta^{*})+(1-\alpha_k)f(\theta_k),
\end{equation*}
the stopping criteria ,
\begin{equation}
    \label{C2}
    {\rm dist}\big(0_{\Tilde{\theta}_k}, \partial h_{\kappa_{cvx}}(\Tilde{\theta}_k, \vartheta_k)\big)<\frac{\kappa_{cvx}}{k+1}d_{\mathcal{R}}(\Tilde{\theta}_k, \vartheta_k),
\end{equation}
  i.e. $\|v_k\|<\frac{\kappa_{cvx}}{k+1}d_{\mathcal{R}}(\Tilde{\theta}_k, \vartheta_k),$ and $\vartheta_k=\mathcal{R}_{\theta_{k-1}}\alpha_{k}\mathcal{R}^{-1}_{\theta_{k-1}}\Tilde{\vartheta}_{k-1},$ and $\Tilde{\vartheta}_k=\mathcal{R}_{\theta_{k-1}}\frac{1}{\alpha_k}\mathcal{R}^{-1}_{\theta_{k-1}}\Tilde{\theta}_k,$ one has 
  \begin{equation*}
    \begin{split}
        f(\theta_k)\leq&\alpha_k f(\theta^{*})+(1-\alpha_k)f(\theta_k)\\
        &+\frac{\kappa_{cvx}\alpha_k^2}{2}\Big(K_1^2\|\mathcal{R}^{-1}_{\theta_{k-1}}\theta^*-\mathcal{R}^{-1}_{\theta_{k-1}}\Tilde{\vartheta}_{k-1}\|^2\\
        &-K_1^4K_2^2\|\mathcal{R}^{-1}_{\theta_{k-1}}\Tilde{\vartheta}_k-\mathcal{R}^{-1}_{\theta_{k-1}}\theta^*\|^2\Big)\\
        &-\frac{\kappa_{cvx}}{2}d_{\mathcal{R}}^2(\Tilde{\theta}_k, \vartheta_k)+\frac{\kappa_{cvx}}{k+1}d_{\mathcal{R}}(\Tilde{\theta}_k, \vartheta_k)  \|\mathcal{R}_{\tilde{\theta}_k}^{-1}\theta\|\\
  \leq&\alpha_k f(\theta^{*})+(1-\alpha_k)f(\theta_k)\\
        &+\frac{\kappa_{cvx}\alpha_k^2}{2}\Big(K_1^2K_2^2d_{\mathcal{R}}^2(\theta^*, \Tilde{\vartheta}_{k-1})-K_1^2K_2^2d_{\mathcal{R}}^2(\theta^{*}, \Tilde{\vartheta}_k)\Big)\\
        &-\frac{\kappa_{cvx}}{2}d_{\mathcal{R}}^2(\Tilde{\theta}_k, \vartheta_k)+\frac{\kappa_{cvx}}{k+1}d_{\mathcal{R}}(\Tilde{\theta}_k, \vartheta_k)d_{\mathcal{R}}(\tilde{\theta}_k, \theta)\\
  \leq&\alpha_k f(\theta^{*})+(1-\alpha_k)f(\theta_k)\\
        &+\frac{\kappa_{cvx}\alpha_k^2}{2}\Big(K_1^2K_2^2d_{\mathcal{R}}^2(\theta^*, \Tilde{\vartheta}_{k-1})-K_1^2K_2^2d_{\mathcal{R}}^2(\theta^{*}, \Tilde{\vartheta}_k)\Big)\\
        &-\frac{\kappa_{cvx}}{2}d_{\mathcal{R}}^2(\Tilde{\theta}_k, \vartheta_k)\\
        &+\frac{\kappa_{cvx}K_1}{k+1}d_{\mathcal{R}}(\Tilde{\theta}_k, \vartheta_k)\|\mathcal{R}_{\theta_{k-1}}^{-1}\tilde{\theta}_k-\mathcal{R}_{\theta_{k-1}}^{-1}\theta\|\\
     =&\alpha_k f(\theta^{*})+(1-\alpha_k)f(\theta_k)\\
       &+\frac{\kappa_{cvx}\alpha_k^2}{2}\Big(K_1^2K_2^2d_{\mathcal{R}}^2(\theta^*, \Tilde{\vartheta}_{k-1})-K_1^2K_2^2d_{\mathcal{R}}^2(\theta^{*}, \Tilde{\vartheta}_k)\Big)\\
       &-\frac{\kappa_{cvx}}{2}d_{\mathcal{R}}^2(\Tilde{\theta}_k, \vartheta_k)\\
       &+\frac{\kappa_{cvx}\alpha_kK_1}{k+1}d_{\mathcal{R}}(\Tilde{\theta}_k, \vartheta_k)\|\mathcal{R}_{\theta_{k-1}}^{-1}\tilde{\vartheta}_k-\mathcal{R}_{\theta_{k-1}}^{-1}\theta^{*}\|\\
 \leq&\alpha_k f(\theta^{*})+(1-\alpha_k)f(\theta_k)\\
       &+\frac{\kappa_{cvx}\alpha_k^2}{2}\Big(K_1^2K_2^2d_{\mathcal{R}}^2(\theta^*, \Tilde{\vartheta}_{k-1})-K_1^2K_2^2d_{\mathcal{R}}^2(\theta^{*}, \Tilde{\vartheta}_k)\Big)\\
       &-\frac{\kappa_{cvx}}{2}d_{\mathcal{R}}^2(\Tilde{\theta}_k, \vartheta_k)\\
       &+\frac{\kappa_{cvx}\alpha_kK_1K_2}{k+1}d_{\mathcal{R}}(\Tilde{\theta}_k, \vartheta_k)d_{\mathcal{R}}(\tilde{\vartheta}_k, \theta^{*}).
    \end{split}
\end{equation*}
So 
\begin{equation}\label{T-3}
\begin{split}
    f(\theta_k)\leq&\alpha_k f(\theta^{*})+(1-\alpha_k)f(\theta_k)\\
                          &+\frac{\kappa_{cvx}\alpha_k^2}{2}\Big(K_1^2K_2^2d_{\mathcal{R}}^2(\theta^*, \Tilde{\vartheta}_{k-1})-K_1^2K_2^2d_{\mathcal{R}}^2(\theta^{*}, \Tilde{\vartheta}_k)\Big)\\
                          &-\frac{\kappa_{cvx}}{2}d_{\mathcal{R}}^2(\Tilde{\theta}_k, \vartheta_k)\\
                          &+\frac{\kappa_{cvx}\alpha_kK_1K_2}{k+1}d_{\mathcal{R}}(\Tilde{\theta}_k, \vartheta_k)d_{\mathcal{R}}(\theta^{*}, \Tilde{\vartheta}_k).
\end{split}
\end{equation}
Set $\mu_k=\frac{1}{k+1}.$ Completing the square yields
\begin{align*}
    -\frac{\kappa_{cvx}}{2}d_{\mathcal{R}}^2(\Tilde{\theta}_k, \vartheta_k)+\kappa_{cvx}\alpha_k\mu_kK_1K_2 d_{\mathcal{R}}(\Tilde{\theta}_k, \vartheta_k)d_{\mathcal{R}}(\theta^{*}, \Tilde{\vartheta}_k)\\
    \leq
    \frac{K_1^2K_2^2\kappa_{cvx}\alpha_k^2\mu_k^2}{2}d_{\mathcal{R}}^2(\theta^{*}, \Tilde{\vartheta}_k),
\end{align*}
and subtracting $f^{*}=f(\theta^{*})$ from both sides, we obtain
\begin{equation*}
\begin{split} 
    f(\theta_k)-f^{*}&\leq(1-\alpha_k)(f(\theta_{k-1})-f^{*})+\frac{\kappa_{cvx}\alpha_k^2}{2}\Big(K_1^2K_2^2d_{\mathcal{R}}^2(\theta^*, \Tilde{\vartheta}_{k-1})\\
    &\quad-K_1^2K_2^2d_{\mathcal{R}}^2(\theta^{*}, \Tilde{\vartheta}_k)\Big)+\frac{K_1^2K_2^2\kappa_{cvx}\alpha_k^2\mu_k^2}{2}d_{\mathcal{R}}^2(\theta^{*}, \Tilde{\vartheta}_k)\\
    &=(1-\alpha)(f(\theta_{k-1})-f^{*})+\frac{\kappa_{cvx}\alpha_k^2K_1^2K_2^2}{2}d_{\mathcal{R}}^2(\theta^*, \Tilde{\vartheta}_{k-1})\\
    &\quad-\frac{\kappa_{cvx}\alpha_k^2K_1^2K_2^2}{2}(1-\mu_k^2)d_{\mathcal{R}}^2(\theta^{*}, \Tilde{\vartheta}_k).
\end{split}
\end{equation*}
So one can obtain
\begin{align*}
    \frac{f(\theta_k)-f^{*}}{\alpha_k^2}+\frac{\kappa_{cvx} K_1^2K_2^2}{2}(1-\mu_k^2)d_{\mathcal{R}}^2(\theta^{*}, \Tilde{\vartheta}_k)\\
    \leq\frac{1-\alpha_k}{\alpha_k^2}(f(\theta_{k-1})-f^{*})+\frac{\kappa_{cvx} K_1^2K_2^2}{2}d_{\mathcal{R}}^2(\theta^*, \Tilde{\vartheta}_{k-1}).
\end{align*}
Denote $A_k=(1-\mu_k^2).$ Using the equality $\frac{1-\alpha_k}{\alpha_k^2}=\frac{1}{\alpha_{k-1}^2}$ we derive the following recursion 
\begin{equation*}
    \begin{split}
         &\frac{f(\theta_k)-f^{*}}{\alpha_k^2}+\frac{\kappa_{cvx} K_1^2K_2^2 A_k}{2}d_{\mathcal{R}}^2(\theta^{*}, \Tilde{\vartheta}_k) \\
         &\leq\frac{1-\alpha_k}{\alpha_k^2}(f(\theta_{k-1})-f^{*})+\frac{\kappa_{cvx} K_1^2K_2^2}{2}d_{\mathcal{R}}^2(\theta^*, \Tilde{\vartheta}_{k-1})\\
         &=\frac{f(\theta_{k-1})-f^{*}}{\alpha_{k-1}^2}+\frac{\kappa_{cvx} K_1^2K_2^2}{2}d_{\mathcal{R}}^2(\theta^{*}, \Tilde{\vartheta}_{k-1})\\
         &\leq\frac{f(\theta_{k-1})-f^{*}}{A_{k-1}\alpha_{k-1}^2}+\frac{\kappa_{cvx} K_1^2K_2^2}{2}d_{\mathcal{R}}^2(\theta^{*}, \Tilde{\vartheta}_{k-1})\\ 
         &=\frac{1}{A_{k-1}}\Bigg(\frac{f(\theta_{k-1})-f^{*}}{\alpha_{k-1}^2}+\frac{\kappa_{cvx} K_1^2K_2^2A_{k-1}}{2}d_{\mathcal{R}}^2(\theta^{*}, \Tilde{\vartheta}_{k-1})\Bigg).
    \end{split}
\end{equation*}
The last inequality holds because $0<A_k\leq1.$
 Iterating $N$  times, we deduce
\begin{align*}
    \frac{f(\theta_N)-f^{*}}{\alpha_N^2}&\leq\frac{f(\theta_N)-f^{*}}{\alpha_N^2}+\frac{\kappa_{cvx} K_1^2K_2^2 A_k}{2}d_{\mathcal{R}}^2(\theta^{*}, \Tilde{\vartheta}_k)\\
    &\leq\frac{\kappa_{cvx} K_1^2K_2^2}{2}d_{\mathcal{R}}^2(\theta^{*}, \theta_0)\prod^N_{k=2}\frac{1}{A_{k-1}}.
\end{align*}
Note that
\begin{equation*}
    \prod^{N}_{k=2}\frac{1}{A_{k-1}}\leq2;
\end{equation*}
thereby with inequality from Lemma 1 we conclude
\begin{align*}
       f(\theta_N)-f^{*}&\leq\frac{\alpha_N^2\kappa_{cvx} K_1^2K_2^{2}}{2}d_{\mathcal{R}}^2(\theta^{*}, \theta_0)\prod^N_{k=2}\frac{1}{A_{k-1}}\\ 
       &\leq\alpha_N^2\kappa_{cvx} K_1^2K_2^2 d_{\mathcal{R}}^2(\theta^{*}, \theta_0)\\
       &\leq\frac{4\kappa_{cvx} K_1^2K_2^2}{(N+1)^2}d_{\mathcal{R}}^2(\theta^{*}, \theta_0).
\end{align*}
Hence
\begin{equation*}
   f(\theta_N)-f^{*}\leq\frac{4\kappa_{cvx} K_1^2K_2^2}{(N+1)^2}d_{\mathcal{R}}^2(\theta^{*}, \theta_0).
\end{equation*}

\end{proof}

\subsection{Strong convexity of the objective function in estimating the intrinsic Fr\'echet means on the sphere}

We provide a proof that the objective functions in estimating both the intrinsic and extrinsic Fr\'echet means on the sphere in Section 4 is strongly convex. 

\begin{proof}
In order to prove the strong-convexity of the intrinsic mean on the sphere $S^n$, we will prove the strong-convexity of the square intrinsic distance function from the point $x_0\in S^n$
\begin{equation*}
    d_g^2(x_0, x)=\arccos^2(x_0^Tx).
\end{equation*}
So for the geodesic from the point $x_1\in S^n$ to the point $x_2\in S^2$
\begin{equation*}
\begin{split}
    \gamma(\lambda)&=\exp_{x_1}\lambda\log_{x_1}x_2\\
    &=\cos\big(\lambda\arccos(x_1^Tx_2)\big)x_1\\
    &\;\;+\sin\big(\lambda\arccos(x_1^Tx_2)\big)\frac{x_2-(x_1^Tx_2)x_1}{\sqrt{1-(x_1^Tx_2)^2}},
\end{split}
\end{equation*}
we need to show following inequality
\begin{align*}
    d_g^2(x_0, \gamma(\lambda))&\leq(1-\lambda)d_g^2(x_0, x_1)+\lambda d_g^2(x_0, x_2)\\
    &\;\;-\frac{\lambda(1-\lambda)\mu}{2}d_g^2(x_1, x_2).
\end{align*}
For the sake of briefness let's use the following notations
\begin{gather*}
d_1=\arccos(x_0^Tx_1), \qquad d_2=\arccos(x_0^Tx_2), \\
d_2=\arccos(x_1^Tx_2).
\end{gather*}

Therefore we have to prove the following inequality
\begin{multline}
\label{strong-con}
    \arccos^2\bigg(\cos(\lambda d_3)\cos(d_1)\\
    +\sin(\lambda d_3)\frac{\cos(d_2)-\cos(d_3)\cos(d_1)}{\sin(d_3)}\bigg)\\
    \leq\quad(1-\lambda)d_1^2+\lambda d_2^2-\frac{\lambda(1-\lambda)\mu}{2}d_3^2.
\end{multline}

Or we should prove the inequality
\begin{align*}
    \arccos^2(x_0^Tx_2)&>\arccos^2(x_0^Tx_1)-2(\log_{x_1}x_0)^T\log_{x_1}x_2\\
    &\;\;+\frac{\mu}{2}\arccos^2(x_1^Tx_2)
\end{align*}
    \begin{align*}
    d_2^2&>d_1^2-2\bigg(d_1\frac{x_0-\cos(d_1)x_1}{\sin(d_1)}\bigg)^T\bigg(d_3\frac{x_2-\cos(d_3)x_1}{\sin(d_3)}\bigg)+\frac{\mu}{2}d_3^2\\
    &=d_1^2-2d_1d_3\frac{\cos(d_2)-\cos(d_3)\cos(d_1)}{\sin(d_1)\sin(d_3)}+\frac{\mu}{2}d_3^2
\end{align*}
The last inequality was checked to hold in Wolfram Mathematica for $d_1,d_2\in[0, \pi/4]$ and $d_3\in\Big[|d_1-d_2|, d_1+d_2\Big],$ where $\mu=1$.

In order to proof the strong-convexity of Fr\'echet function in estimating extrinsic mean on the sphere $S^n$, we will prove the strong-convexity of the square extrinsic distance function from the point $x_0\in S^n$
\begin{equation*}
    d_e^2(x_0, x)=2(1-x_0^Tx).
\end{equation*}
So for the geodesic from the point $x_1\in S^n$ to the point $x_2\in S^2$
\begin{equation*}
\begin{split}
    \gamma(\lambda)&=\exp_{x_1}\lambda\log_{x_1}x_2\\
    &=\cos\big(\lambda\arccos(x_1^Tx_2)\big)x_1+\sin\big(\lambda\arccos(x_1^Tx_2)\big)\frac{x_2-(x_1^Tx_2)x_1}{\sqrt{1-(x_1^Tx_2)^2}},
\end{split}
\end{equation*}
we need to show that
\begin{align*}
    d_e^2(x_0, \gamma(\lambda))&\leq(1-\lambda)d_e^2(x_0, x_1)+\lambda d_e^2(x_0, x_2)\\
    &\;\;-\frac{\lambda(1-\lambda)\mu}{2}d_g^2(x_1, x_2).
\end{align*}

Therefore we have to prove 
\begin{multline}\label{strong-con}
    2\bigg(1-\cos(\lambda d_3)\cos(d_1)-\sin(\lambda d_3)\frac{\cos(d_2)-\cos(d_3)\cos(d_1)}{\sin(d_3)}\bigg)\\
    \leq\quad2-2\big((1-\lambda)\cos(d_1)+\lambda\cos( d_2)\big)-\frac{\lambda(1-\lambda)\mu}{2}d_3^2.
\end{multline}

Or we need to show
\begin{multline*}
    2(1-x_0^Tx_2)>\\
    2(1-x_0^Tx_1)-2\big(x_0-(x_0^Tx_1)x_1\big)^T\log_{x_1}x_2+\frac{\mu}{2}\arccos^2(x_1^Tx_2)
\end{multline*}
Thus
    \begin{align*}
    &2\big(1-\cos(d_2)\big)\\
    &\quad\;>2\big(1-\cos(d_1)\big)- 2\big(x_0-\cos(d_1)x_1\big)^T\bigg(d_3\frac{x_2-\cos(d_3)x_1}{\sin(d_3)}\bigg)+\frac{\mu}{2}d_3^2\\
    &\quad=2\big(1-\cos(d_1)\big)-2d_3\frac{\cos(d_2)-\cos(d_3)\cos(d_1)}{\sin(d_3)}+\frac{\mu}{2}d_3^2
\end{align*}
The last inequality was verified Wolfram Mathematica for $d_1,d_2\in[0, \pi/4]$ and $d_3\in\Big[|d_1-d_2|, d_1+d_2\Big],$ where $\mu=1$

\end{proof}

\subsubsection*{Acknowledgments}
\vspace{-.5em}
Lizhen Lin would like to thank Dong Quan Nguyen for very helpful discussions.  Lizhen Lin acknowledges the support from NSF grants IIS  1663870, DMS Career 1654579 and a DARPA grant N66001-17-1-4041.   Bayan Saparbayeva was partially supported by DARPA N66001-17-1-4041. 

\clearpage
\bibliographystyle{apalike}
\bibliography{sample}

\end{document}